\documentclass[11pt]{article}

\usepackage{fullpage}

\usepackage{makecell}
\usepackage{array}
\usepackage{comment}
\usepackage{epsf}
\usepackage{subcaption}
\usepackage{graphicx}
\usepackage{bbm}
\usepackage{color}
\usepackage{nicefrac}
\usepackage{tcolorbox}
\usepackage{mathtools}
\usepackage{amsfonts}
\usepackage{amsthm}
\usepackage{amsmath, bm}
\usepackage{amssymb}
\usepackage{textcomp}
\usepackage{lipsum}
\usepackage{xcolor}
\usepackage{graphicx}
\usepackage{wrapfig}
\usepackage{mathtools}
\usepackage{enumitem}
\usepackage{xspace}

\usepackage{algorithm}
\usepackage{algpseudocode}
\counterwithin{algorithm}{section}

\usepackage{tikz}
\usetikzlibrary{shapes,snakes}
\usetikzlibrary{arrows.meta}
\usetikzlibrary{decorations.pathmorphing}
\usetikzlibrary{patterns}
\usetikzlibrary{calc}

\usepackage{amsbsy}

\usepackage[linktocpage=true]{hyperref}
\definecolor{OliveGreen}{HTML}{3C8031}
\hypersetup{
colorlinks=true,
urlcolor=blue,
linkcolor=blue,
citecolor=OliveGreen,
}

\usepackage{xurl}
\hypersetup{breaklinks=true}
\usepackage[hyphenbreaks]{breakurl}
\newenvironment{fminipage}%
  {\end{minipage}\end{Sbox}\fbox{\TheSbox}}

\newtheorem{theorem}{Theorem}[section]

\newtheorem{proposition}[theorem]{Proposition}
\newtheorem*{proposition*}{Proposition}
\newtheorem{lemma}[theorem]{Lemma}

\newtheorem{question}{Question}

\theoremstyle{definition}
\newtheorem{definition}{Definition}[section]



\usepackage[backend=biber, style=alphabetic, sorting=nyt, maxnames=100,backref=true, firstinits=true, sortcites = true]{biblatex}
\newlength{\widebarargwidth}
\newlength{\widebarargheight}
\newlength{\widebarargdepth}

\makeatletter
\long\def\@makecaption#1#2{
       \vskip 0.8ex
       \setbox\@tempboxa\hbox{\small {\bf #1:} #2}
       \parindent 1.5em 
       \dimen0=\hsize
       \advance\dimen0 by -3em
       \ifdim \wd\@tempboxa >\dimen0
               \hbox to \hsize{
                       \parindent 0em
                       \hfil 
                       \parbox{\dimen0}{\def\baselinestretch{0.96}\small
                               {\bf #1.} #2
                               } 
                       \hfil}
       \else \hbox to \hsize{\hfil \box\@tempboxa \hfil}
       \fi
       }
\makeatother

\long\def\comment#1{}

\newcommand{\set}[1]{\{#1\}}
\newcommand{\defeq}{\coloneqq}
\newcommand{\N}{\mathbb{N}}

\newcommand{\R}{\mathbb{R}}
\newcommand{\E}{\mathbb{E}}

\newcommand{\poly}{\mathsf{poly}}

\renewcommand{\epsilon}{\varepsilon}
\newcommand{\eps}{\epsilon}

\newcommand{\dom}{\mathsf{dom}}
\newcommand{\0}{\varnothing}
\newcommand{\bemph}[1]{{\normalfont#1}} % define how emphasised brackets should look
\newcommand{\ep}[1]{\bemph{(}#1\bemph{)}} % parentheses

\newcommand{\pto}{\dashrightarrow}
\newcommand{\emphdef}[1]{\textbf{\textit{{#1}}}}
\newcommand{\emphd}[1]{\emphdef{#1}}
\newcommand{\blank}{\mathsf{blank}}

\newcommand{\im}{\mathrm{im}}
\newcommand{\Bad}{\mathsf{Bad}}
\newcommand{\Good}{\mathsf{Good}}
\newcommand{\keptedges}{E_{K}(v)}
\newcommand{\uncoloredges}{E_{U}(v)}
\newcommand{\normalizeddeg}{d_{K\cap U}(v)}
\newcommand{\nd}{\normalizeddeg}
\newcommand{\keep}{\mathsf{keep}}
\newcommand{\uncolor}{\mathsf{uncolor}}
\newcommand{\eq}{\mathsf{eq}}

\numberwithin{equation}{section}
\newcommand{\dlll}{\mathfrak{d}}

\title{Palette Sparsification for Graphs with Sparse Neighborhoods}

\usepackage{authblk}

\author{Abhishek Dhawan\thanks{Email: adhawan2@illinois.edu.}}
\affil{Department of Mathematics, University of Illinois Urbana--Champaign}

\date{}

\begin{document}

\maketitle

\begin{abstract}
    A seminal \textit{palette sparsification} result of Assadi, Chen, and Khanna states that in every $n$-vertex graph of maximum degree $\Delta$, sampling $\Theta(\log n)$ colors per vertex from $\set{1, \ldots, \Delta+1}$ almost certainly allows for a proper coloring from the sampled colors. Alon and Assadi extended this work proving a similar result for $O\left(\Delta/\log \Delta\right)$-coloring triangle-free graphs. Apart from being interesting results from a combinatorial standpoint, their results have various applications to the design of graph coloring algorithms in different models of computation.

    In this work, we focus on \textit{locally sparse graphs}, i.e., graphs with sparse neighborhoods. We say a graph $G = (V, E)$ is $k$-locally-sparse if for each vertex $v \in V$, the subgraph $G[N(v)]$ contains at most $k$ edges. A celebrated result of Alon, Krivelevich, and Sudakov shows that such graphs are $O(\Delta/\log (\Delta/\sqrt{k}))$-colorable. For any $\alpha \in (0, 1)$ and $k \ll \Delta^{2\alpha}$, let $G$ be a $k$-locally-sparse graph. For $q = \Theta\left(\Delta/\log \left(\Delta^\alpha/\sqrt{k}\right)\right)$, we show that
    sampling $O\left(\Delta^\alpha + \sqrt{\log n}\right)$ colors per vertex is sufficient to obtain a proper $q$-coloring of $G$ from the sampled colors. Setting $k = 1$ recovers the aforementioned result of Alon and Assadi for triangle-free graphs.

    We establish this result in the more general setting of correspondence coloring where the structural constraints are placed on the correspondence cover as opposed to on the underlying graph. Such a notion has gathered much interest in recent years. A key element in our proof is a proposition which determines sufficient conditions on the correspondence cover $\mathcal{H}$ for $G$ to admit a proper $\mathcal{H}$-coloring. This result improves upon recent work of Anderson, Kuchukova, and the author, and is of independent interest.
\end{abstract}

\newpage

\sloppy

\section{Introduction}\label{section:intro}

\subsection{Background and Results}\label{subsection: background}

Let $G = (V, E)$ be an $n$-vertex graph of maximum degree $\Delta$.
For $q \in \N$, a \emphd{proper $q$-coloring} of $G$ is a function $\phi\,:\,V \to [q]$ (here, $[q] = \set{1, \ldots, q}$) such that $\phi(u) \neq \phi(v)$ whenever $uv \in E$.
The \emphd{chromatic number} of $G$ (denoted $\chi(G)$) is the minimum value $q$ for which $G$ admits a proper $q$-coloring.
Introduced independently by Vizing \cite{vizing1976coloring} and Erd\H{o}s, Rubin, and Taylor \cite{erdos1979choosability}, list coloring is a generalization of graph coloring in which each vertex is assigned a color from its own predetermined list of colors.
We define list coloring formally as follows (for a set $X$, $2^X$ denotes the power set of $X$):

\begin{definition}[List Coloring]\label{defn: list coloring}
    Let $G = (V, E)$ be a graph.
    A \emphd{list assignment} for $G$ is a function $L \,:\, V \to 2^\N$.
    A \emphd{proper $L$-coloring} is a function $\phi\,:\, V \to \N$ such that $\phi(v) \in L(v)$ for each $v \in V$ and $\phi(u) \neq \phi(v)$ for each edge $uv \in E$.
    We call a list assignment \emphd{$q$-fold} if $|L(v)| \geq q$ for each $v \in V$.
    The \emphd{list chromatic number} of $G$ (denoted $\chi_\ell(G)$) is the minimum value $q$ for which $G$ admits a proper $L$-coloring for every $q$-fold list assignment $L$ for $G$.
\end{definition}

It is well-known that $\chi(G) \leq \Delta + 1$.
In fact, such a coloring can be obtained by a simple greedy algorithm with running time $O(n\,\Delta)$: simply pick uncolored vertices in an arbitrary order and assign a color to a vertex not yet assigned to any of its neighbors; since the maximum degree is $\Delta$, such a color always exists.
When processing dense graphs, even this simple algorithm can be computationally prohibitive. 
This is due to various limitations which arise in processing massive graphs such as requiring to process the graph in a streaming fashion, on a single machine, in parallel across multiple machines due to storage limitations, or simply not having enough time to read the entire input.
This motivates the design of \textit{sublinear algorithms}, i.e., algorithms which require computational resources that are substantially smaller than the size of their input.
Assadi, Chen, and Khanna \cite{assadi2019sublinear} first studied sublinear algorithms for the $(\Delta + 1)$-coloring problem.
The following technical result is a key part of their proof.

\begin{theorem}[{\cite[Theorem 1]{assadi2019sublinear}}]\label{theorem: assadi chen khanna}
    Let $G = (V, E)$ be an $n$-vertex graph of maximum degree $\Delta$.
    Define a list assignment $L$ for $G$ by sampling $L(v)$ independently for each vertex $v \in V$ from the $\Theta(\log n)$-size subsets of $[\Delta + 1]$.
    Then, $G$ admits a proper $L$-coloring with high probability\footnote{Throughout this paper, ``with high probability'' means with probability at least $1 - 1/\poly(n)$.}.
\end{theorem}

Theorem~\ref{theorem: assadi chen khanna} is an example of a so-called \emphd{palette sparsification} result.
Assadi, Chen, and Khanna show how a palette sparsification result can be used to define efficient sublinear algorithms \cite{assadi2019sublinear}.
The idea of their method, in a nutshell, is as follows. 
Suppose we are trying to compute a proper $q$-coloring of a graph $G = (V, E)$, but, due to limited computational resources, we cannot keep track of all the edges of $G$ at once. 
Let us independently sample, for each vertex $v \in V$, a random set $L(v)$ of colors of size $|L(v)| = s \ll q$. 
Define
\[
    E_\mathrm{conflict} \,\defeq\, \set{uv \in E \,:\, L(u) \cap L(v) \neq \0}.
\]
If we color $G$ by assigning to every vertex $v \in V$ a color from the corresponding set $L(v)$, then only the edges in $E_\mathrm{conflict}$ may become monochromatic, so instead of working with the entire edge set of $G$, we only need to keep track of the edges in the \ep{potentially much smaller} set $E_\mathrm{conflict}$. 
For this strategy to succeed, we must ensure that, with high probability, it is indeed possible to properly color $G$ using the colors from the sets $L(v)$.
Through standard probabilistic techniques, one can show that for $q = \Delta + 1$ and $s = \Theta(\log n)$, the total number of edges in $E_{\mathrm{conflict}}$ is with high probability only $\tilde O(n)$\footnote{Here, and in what follows, $\tilde O(x) = O(x\,\poly(\log x))$}.
Theorem~\ref{theorem: assadi chen khanna} thus reduces the $(\Delta + 1)$-coloring problem to a list-coloring problem on a graph with much fewer edges.

In \cite{KahnKenney}, Kahn and Kenney tightened the result of Theorem~\ref{theorem: assadi chen khanna} by showing that sampling $(1+o(1))\log n$ colors per vertex is enough (which is best possible). 
The literature on the subject of palette sparsification is quite extensive, especially since notions closely related to it have also been studied independently of their algorithmic applications, with a particular attention in recent years directed to the case of list-\emph{edge}-coloring from random lists. 
For a sample of results concerning palette sparsification from the algorithmic perspective, see \cite{Sparsification1, alon2020palette, Sparsification2, Sparsification3}; for a more purely probabilistic focus, see, e.g., \cite{KN1,KN2,C1,C2,C3,CH,C4,rand4,rand3,rand1,rand2, anderson2022coloring}.

The approach of Assadi, Chen, and Khanna discussed earlier yields a framework for applying palette sparsification theorems in a black-box manner to develop $q$-coloring algorithms (for appropriate values of $q$ and $s$) in a variety of classical models of sublinear computation, including (dynamic) graph streaming algorithms, sublinear time/query algorithms, and massively parallel computation (MPC) algorithms.
Motivated by Theorem~\ref{theorem: assadi chen khanna} and the aforementioned algorithmic implications, a natural question to ask is the following: 
\begin{quote}
    \textit{What other graph coloring problems admit similar palette sparsification theorems?}
\end{quote}
This is the main focus of this paper.

A central area of research in graph coloring aims to determine the chromatic number of special families of graphs.
A celebrated result of Johansson \cite{Joh_triangle} provides the bound $\chi(G) = O(\Delta/\log \Delta)$ for triangle-free graphs (see a simplified proof with a constant factor of $1 + o(1)$ in \cite{Molloy}; see also \cite{bernshteyn2019johansson, PS15, alon2020palette}).
Alongside several other palette sparsification theorems, Alon and Assadi prove the following result for $O(\Delta/\log \Delta)$-coloring triangle-free graphs:

\begin{theorem}[{\cite[Theorem 2]{alon2020palette}}]\label{theorem: alon assadi}
    Let $G = (V, E)$ be an $n$-vertex triangle-free graph of maximum degree $\Delta$, and let $\alpha \in (0, 1)$ and $q = 9\Delta/\log\left(\Delta^\alpha\right)$.
    Define a list assignment $L$ for $G$ by sampling $L(v)$ independently for each vertex $v \in V$ from the $\Theta\left(\Delta^\alpha + \sqrt{\log n}\right)$-size subsets of $[q]$.
    Then, $G$ admits a proper $L$-coloring with high probability.
\end{theorem}

In a similar flavor, Anderson, Bernshteyn, and the author considered the family of $K_{1, s, t}$-free graphs \cite{anderson2022coloring}, where $K_{1, s, t}$ is the complete tripartite graph with partitions of size $1$, $s$, and $t$.
As a corollary to their main theorem, they obtain a palette sparsification result.
For $s = t = 1$, the result recovers that of Theorem~\ref{theorem: alon assadi} with an improved constant factor in the definition of $q$.

\begin{theorem}[{\cite[Corollary 1.16]{anderson2022coloring}}]\label{theorem: K1tt palette}
    Let $G = (V, E)$ be an $n$-vertex $K_{1, s, t}$-free graph of maximum degree $\Delta$, and let $\alpha \in (0, 1)$ and $q = (4 + o(1))\Delta/\log\left(\Delta^\alpha\right)$.
    Define a list assignment $L$ for $G$ by sampling $L(v)$ independently for each vertex $v \in V$ from the $\Theta\left(\Delta^\alpha + \sqrt{\log n}\right)$-size subsets of $[q]$.
    Then, $G$ admits a proper $L$-coloring with high probability.
\end{theorem}

We remark that the above result holds in the more general setting of correspondence coloring with structural constraints placed on the correspondence cover (see \S\ref{subsection: list and correspondence coloring} for an overview of this setting).
Building upon the works mentioned so far, we consider the family of \emphd{locally sparse} graphs, which we formally define as follows:

\begin{definition}[Local Sparsity]
    For $k \in \R$, a graph $G = (V, E)$ is \emphd{$k$-locally-sparse} if for each $v \in V$, the subgraph $G[N(v)]$ contains at most $k$ edges.
\end{definition}

Alon, Krivelevich, and Sudakov famously showed that the chromatic number of a $k$-locally-sparse graph is $O(\Delta/\log (\Delta/\sqrt{k}))$ \cite{AKSConjecture} (see \cite{Vu} for a proof in the list coloring setting; see also \cite[\S5.3]{DKPS} for a proof in the correspondence coloring setting with an improved constant factor of $1 + o(1)$).
Our main result is a palette sparsification theorem for the $O(\Delta/\log (\Delta/\sqrt{k}))$-coloring problem for $k$-locally-sparse graphs.
We remark that the result does not hold for all $1 \leq k \leq \binom{\Delta}{2}$, however, it does hold for a large range.

\begin{theorem}\label{theo: main theorem}
    Let $\eps > 0$ and $\gamma \in (0, 1)$ be arbitrary.
    There exists a constant $C > 0$ such that for all $\alpha \in (0, 1)$, there is $\Delta_0 \in \N$ with the following property.
    Let $\Delta\geq \Delta_0$, $n \in \N$, and let $k$ and $q \geq s$ be such that
    \[1\,\leq\, k \,\leq\, \Delta^{2\alpha\gamma}, \qquad q \defeq \dfrac{4(1+\gamma + \eps)\Delta}{\log \left(\Delta^\alpha/\sqrt{k}\right)}, \qquad \text{and} \qquad s \geq \Delta^\alpha + C\sqrt{\log n}.\]
    Let $G = (V, E)$ be an $n$-vertex $k$-locally-sparse graph of maximum degree $\Delta$.
    Suppose for each vertex $v \in V$, we independently sample a set $L(v) \subseteq [q]$ of size $s$ uniformly at random.
    Then, with probability at least $1 - 1/n$ there exists a proper coloring $\phi$ of $G$ such that $\phi(v) \in L(v)$ for each $v \in V$.
\end{theorem}

This is the first palette sparsification theorem for this problem.
Additionally, as in the case of Theorem~\ref{theorem: K1tt palette}, Theorem~\ref{theo: main theorem} holds in the more general setting of correspondence coloring with structural constraints placed on the correspondence cover (see Theorem~\ref{theo: main theorem dp color degree}).
Furthermore, by setting $k = 1$ and letting $\gamma$ be arbitrarily small, we recover the result of Theorem~\ref{theorem: K1tt palette} for $s = t = 1$, which improves the constant factor in the definition of $q$ in Theorem~\ref{theorem: alon assadi}.

A key step in the proof of Theorem~\ref{theo: main theorem} is a proposition involving list coloring, which is of independent interest.
Before we state the result, we make a few definitions for a graph $G = (V, E)$ and a list assignment $L$ for $G$.
For a vertex $v \in V$ and a color $c \in L(v)$, we let
\[N_L(v, c) \defeq \set{u \in N(v)\,:\, c\in L(u)}, \qquad \text{and} \qquad \deg_L(v, c) \defeq |N_L(v, c)|.\]
We refer to $\deg_L(v, c)$ as the \emphd{$c$-degree} of $v$.
Let us now state our result.

\begin{proposition}\label{prop: main prop}
    Let $\gamma \in (0, 1)$ be arbitrary.
    There exists a constant $d_0$ such that the following holds for all $d\geq d_0$ and $1 \leq k \leq d^{2\gamma}$.
    Suppose $G = (V, E)$ is a graph and $L \,:\, V \to 2^{\N}$ is a list assignment for $G$ such that:
    \begin{enumerate}[label = (\textbf{S\arabic*})]
        \item\label{item: list color degree} for every vertex $v \in V$ and color $c \in L(v)$, $\deg_L(v, c) \leq d$,
        \item\label{item: list local sparsity} for every vertex $v \in V$ and color $c \in L(v)$, $|E\left(G[N_L(v, c)]\right)| \leq k$, and
        \item\label{item: list size not DP} for every $v \in V$, $|L(v)| \geq \dfrac{4(1 + \gamma)d}{\log \left(d/\sqrt{k}\right)}$.
    \end{enumerate}
    Then, $G$ admits a proper $L$-coloring.
\end{proposition}

This framework, with constraints on $c$-degrees instead of the maximum degree, is often referred to as the \emphdef{color--degree} setting.
Pioneered by Kahn \cite{KahnListEdge}, Kim \cite{Kim95}, Johansson \cite{Joh_triangle,Joh_sparse}, and Reed \cite{Reed}, among others, this setting is considered the benchmark for list coloring.
(For a selection of a few more recent examples, see \cite{BohmanHolzman, ReedSud, LohSudakov, alon2020palette, cambie2022independent, KangKelly, GlockSudakov, anderson2022coloring, anderson2023colouring, cambie2022independent, anderson2024coloring, dhawan2023list, alon2021asymmetric}.)
We remark that Proposition~\ref{prop: main prop} holds in the more general setting of correspondence coloring (see Proposition~\ref{prop: dp coloring locally sparse}).

The rest of this introduction is structured as follows:
% in \S\ref{subsection: main results}, we formally state our results; 
in \S\ref{subsection: proof overview}, we provide an overview of the proof ideas; in \S\ref{subsection: future}, we outline potential future directions of inquiry.

\subsection{Proof Overview}\label{subsection: proof overview}

We first discuss the proof of Proposition~\ref{prop: main prop}.
As mentioned earlier, it plays a key role in the proof of our main results.
The proof is rather involved and constitutes the bulk of our effort in this paper (see \S\ref{section: wcp} for the details.)

\paragraph{List-Coloring with Constraints on Color--Degrees.}
We will employ a variant of the so-called ``R\"odl nibble method'' in order to prove Proposition~\ref{prop: main prop}.
Such a coloring procedure was first considered by Kim \cite{Kim95}, based on a technique of R\"odl in \cite{RODL198569}. 
Molloy and Reed further developed this procedure, which they coined the \hyperref[algorithm: wcp]{Wasteful Coloring Procedure} \cite{MolloyReed}.
This technique was employed by Alon and Assadi in the proof of Theorem~\ref{theorem: alon assadi} and by Anderson, Bernshteyn, and the author in the proof of Theorem~\ref{theorem: K1tt palette}.
Indeed, we apply the same coloring procedure with the analysis tailored to our setting.
The algorithm proceeds in stages.
At each stage, we have a graph $G$ and a list assignment $L\,:\,V(G) \to 2^\N$ for $G$.
We randomly construct a partial $L$-coloring $\phi$ of $G$ and lists $L'(v) \subseteq L(v)$ satisfying certain properties.
In particular, let $G_\phi$ be the graph induced by the vertices in $G$ uncolored under $\phi$.
For a vertex $v \in V(G_\phi)$,
we define $L_\phi(v)$ to be the set of \textit{available colors} for $v$, i.e.,
\[L_\phi(v) \defeq \set{c \in L_\phi(v)\,:\, \forall u \in N_L(v, c),\, \phi(u) \neq c}.\]
Note that $\phi$ can be extended to include $v$ by assigning to $v$ any color in $L_\phi(v)$.
It seems natural that the lists $L'(\cdot)$ produced by the \hyperref[algorithm: wcp]{Wasteful Coloring Procedure} should satisfy $L'(v) = L_{\phi}(v)$; however, in actuality $L'(v)$ is potentially a strict subset of $L_{\phi}(v)$ (hence the ``wasteful'' in the name; see the discussion in \cite[Chapter 12.2.1]{MolloyReed} for the utility of such ``wastefulness'').

Define the following parameters:
\[\ell \defeq \min |L(v)|, \qquad \ell' \defeq \min|L'(v)|, \qquad d \defeq \max_{v, c}\deg_L(v, c), \qquad \text{and} \qquad d' \defeq \max_{v, c}\deg_{L'}(v, c).\]
We say $\phi$ is ``good'' if the ratio $d'/\ell'$ is considerably smaller than $d/\ell$.
If indeed the output of the \hyperref[algorithm: wcp]{Wasteful Coloring Procedure} produces a ``good'' coloring, then we may repeatedly apply the procedure until we are left with a list assignment $\tilde L$ of the uncolored vertices such that $\min |\tilde L(v)| \geq 8\max_{v, c}\deg_{\tilde L}(v, c)$.
At this point, we can complete the coloring by applying the following proposition:

\begin{proposition}\label{prop:final_blow_list}
    Let $G = (V, E)$ be a graph and let $L\,:\, V \to 2^\N$ be a list assignment for $G$. 
    If there is an integer $\ell$ such that $|L(v)| \geq \ell$ and $\deg_L(v, c) \leq \ell/8$ for each $v \in V$ and $c \in L(v)$, then $G$ admits a proper $L$-coloring.
\end{proposition}

This proposition is standard and proved using the \hyperref[LLL]{Lov\'asz Local Lemma}. 
Its first appearance is in the paper \cite{Reed} by Reed; see also \cite[\S4.1]{MolloyReed} for a textbook treatment.
Note that a result of Haxell allows one to replace the constant $8$ with $2$ \cite{haxell2001note}.

The heart of our argument lies in proving that the output of the \hyperref[algorithm: wcp]{Wasteful Coloring Procedure} is ``good''.
The proof follows the approach developed in earlier work of Jamall \cite{Jamall}, Pettie and Su \cite{PS15}, and Alon and Assadi \cite{alon2020palette} on triangle-free graphs, and Anderson, Bernshteyn, and the author on $K_{1, s, t}$-free graphs \cite{anderson2022coloring}. 
Anderson, Kuchukova, and the author adapted this strategy to correspondence coloring locally-sparse graphs (see Theorem~\ref{theorem: local sparsity ADK}).
With a more careful analysis, we extend the range of $k$ considered in their result as well as provide an improved bound on the list sizes.

\paragraph{From List Coloring to Palette Sparsification.}
The approach here follows that of the proofs of Theorems~\ref{theorem: alon assadi} and \ref{theorem: K1tt palette}.
The details are provided in \S\ref{section: palette}.
The goal is to show that the list assignment $L$ sampled in Theorem~\ref{theo: main theorem} satisfies the conditions of Proposition~\ref{prop: main prop} with high probability for an appropriate choice of parameters.
It turns out that showing this is not quite so simple.
Instead, we let 
\[L'(v) \defeq \set{c \in L(v)\,:\, \deg_L(v, c) = \tilde O(\Delta^\alpha)},\]
for each $v \in V$.
We show that the list assignment $L'$ satisfies the conditions of Proposition~\ref{prop: main prop} with high probability for $d \defeq \tilde \Theta(\Delta^\alpha)$.
The color--degree condition \ref{item: list color degree} follows by definition.
Additionally, the local sparsity condition \ref{item: list local sparsity} follows as a result of the following chain of inequalities:
\[|E\left(G[N_{L'}(v, c)]\right)| \,\leq\, |E(G[N(v)])| \,\leq\, k \,\leq\, \Delta^{2\alpha\gamma} \,\leq\, d^{2\gamma}.\]
It remains to show the condition on the list sizes \ref{item: list size not DP} holds with high probability.
It turns out to be enough to show that $|L(v)\setminus L'(v)|$ is not too large for each $v \in V$.
This is achieved by employing a version of Chernoff's inequality tailored to negatively correlated random variables (see Theorem~\ref{lemma:NegativeChernoff}).

\subsection{Concluding Remarks}\label{subsection: future}

In this work, we consider palette sparsification of the $O(\Delta/\log(\Delta/\sqrt{k}))$-coloring problem for $k$-locally-sparse graphs of maximum degree $\Delta$.
Enroute to this goal, we provide improved bounds for list coloring (and correspondence coloring) locally sparse graphs in the color--degree setting, which is of independent interest.
We conclude this introduction with potential future directions of inquiry.

Our results in this paper hold for $k \ll \Delta^{2}$.
Note that for $k = \Theta(\Delta^2)$, in the worst case the graph is complete and the problem reduces to palette sparsification for $(\Delta + 1)$-coloring, which has been extensively studied \cite{assadi2019sublinear, KahnKenney, flin2024distributed, kahn2024asymptotics}.
The proof of Proposition~\ref{prop: main prop} requires that $k \ll d^2$ and as it plays a key role in the proof of Theorem~\ref{theo: main theorem}, it is unclear how to extend the result to the full range of $k$.

\begin{question}
    Can we obtain a palette sparsification result for the $O(\Delta/\log(\Delta/\sqrt{k}))$-coloring problem of $k$-locally-sparse graphs for all $1 \leq k \leq \binom{\Delta}{2}$?
\end{question}

In recent work, the author considered the more general notion of \emphd{$(k, r)$-local-sparsity}.
We say a graph $G = (V, E)$ is $(k, r)$-locally-sparse if for each $v \in V$, the subgraph $G[N(v)]$ contains at most $k$ copies of $K_r$.
The main result of the paper is the following:

\begin{theorem}[{\cite[Theorem 1.7]{dhawan2024bounds}}]\label{theo: loc sparse}
    Let $G$ be a $(k, r)$-locally-sparse graph of maximum degree $\Delta$.
    Then, $\chi(G) = O\left(\frac{\Delta}{\log \Delta}\max\left\{r\log\log \Delta,\, \frac{\log k}{r}\right\}\right)$.
\end{theorem}

For small values of $k$, this matches a result of Johannson for $K_{r+1}$-free graphs \cite{Joh_sparse} (however, the proof of Theorem~\ref{theo: loc sparse} does not yield an efficient algorithm).
It is natural to ask whether one can prove a palette sparsification result in this more general setting.
(We remark that such a result is not known even for $k = 1$.)

\begin{question}
    Can we obtain a palette sparsification result for the Theorem~\ref{theo: loc sparse}?
\end{question}

Alon and Assadi also considered a \textit{local} version of the $(\Delta + 1)$-coloring problem.
Here, each vertex $v \in V$ must be assigned a color in $[\deg(v) + 1]$.
They proved a palette sparsification result for this problem \cite[Theorem 3]{alon2020palette}.
A number of results in the graph coloring literature hold in this local setting, i.e., $\phi(v) \in [f(\deg(v))]$ as opposed to $\phi(v) \in [f(\Delta)]$ for some function $f$ (see, e.g., \cite{davies2020coloring, DKPS, bonamy2022bounding, dhawan2024bounds, kelly2024fractional}).
We ask whether a local version of our main result holds (note that this is not known even for triangle-free graphs).

\begin{question}
    Let $G = (V, E)$ be an $n$-vertex graph.
    For $\alpha \in (0, 1)$, let $k,\,q\,:\,V\to \R$ be defined as follows:
    \[k(v) \defeq |E(G[N(v)])|, \qquad q(v) \defeq O\left(\frac{\deg(v)}{\log\left(\deg(v)^\alpha/\sqrt{k(v)}\right)}\right).\]
    Is sampling $\Theta(\deg(v)^\alpha + \sqrt{\log n})$ colors from $[q(v)]$ for each $v \in V$ sufficient to obtain a proper coloring of $G$?
\end{question}

The rest of this paper is structured as follows: in \S\ref{section: prelim}, we introduce some definitions and probabilistic tools that will be used later in the paper; in \S\ref{section: palette}, we prove Theorem~\ref{theo: main theorem} modulo Proposition~\ref{prop: main prop}, which we prove in \S\ref{section: wcp}.

\section{Preliminaries}\label{section: prelim}

This section is split into two subsections.
In the first, we define correspondence coloring and discuss implications of our results in this setting.
In the second, we describe the probabilistic tools we will employ in our proofs.

\subsection{Correspondence Coloring}\label{subsection: list and correspondence coloring}

\emphd{Correspondence coloring} (also known as \emphd{DP-coloring}) is a generalization of list coloring introduced by Dvo\v{r}\'ak and Postle \cite{DPCol} in order to solve a question of 
Borodin. 
Just as in list coloring, each vertex is assigned a list of colors, $L(v)$;
in contrast to list coloring, though, the identifications between the colors in the lists are allowed to vary from edge to edge.
That is, each edge $uv \in E(G)$ is assigned a matching $M_{uv}$ (not necessarily perfect and possibly empty) from $L(u)$ to $L(v)$.
A proper correspondence coloring is a mapping $\phi : V(G) \to \N$ satisfying $\phi(v) \in L(v)$ for each $v \in V(G)$ and $\phi(u)\phi(v) \notin M_{uv}$ for each $uv \in E(G)$.
Formally, correspondence colorings are defined in terms of an auxiliary graph known as a \emphd{correspondence cover} of $G$.
We take the following definition from \cite[Definition 1.9]{anderson2024coloring}:

\begin{definition}[Correspondence Cover]\label{def:corr_cov}
    A \emphd{correspondence cover} of a graph $G$ is a pair $\mathcal{H} = (L, H)$, where $H$ is a graph and $L \,:\,V(G) \to 2^{V(H)}$ such that:
    \begin{enumerate}[label= \ep{\normalfont CC\arabic*}, leftmargin = \leftmargin + 1\parindent]
        \item The set $\set{L(v)\,:\, v\in V(G)}$ forms a partition of $V(H)$,
        \item\label{dp:list_independent} For each $v \in V(G)$, $L(v)$ is an independent set in $H$, and
        \item\label{dp:matching} For each $u, v\in V(G)$, the edge set of $H[L(u) \cup L(v)]$ forms a matching, which is empty if $uv \notin E(G)$.
    \end{enumerate}
\end{definition}
We call the vertices of $H$ \emphd{colors}.
For $c \in V(H)$, we let $L^{-1}(c)$ denote the \emphd{underlying vertex} of $c$ in $G$, i.e., the unique vertex $v \in V(G)$ such that $c \in L(v)$. If two colors $c$, $c' \in V(H)$ are adjacent in $H$, we say that they \emphd{correspond} to each other and write $c \sim c'$.

An \emphd{$\mathcal{H}$-coloring} is a mapping $\phi \colon V(G) \to V(H)$ such that $\phi(v) \in L(v)$ for all $v \in V(G)$. Similarly, a \emphd{partial $\mathcal{H}$-coloring} is a partial mapping $\phi \colon V(G) \pto V(H)$ such that $\phi(v) \in L(v)$ whenever $\phi(v)$ is defined. A \ep{partial} $\mathcal{H}$-coloring $\phi$ is \emphd{proper} if the image of $\phi$ is an independent set in $H$, i.e., if $\phi(u) \not \sim \phi(v)$
for all $u$, $v \in V(G)$ such that $\phi(u)$ and $\phi(v)$ are both defined. Notice, then, that the image of a proper $\mathcal{H}$-coloring of $G$ is exactly an independent set $I \subseteq V(H)$ with $|I \cap L(v)| = 1$ for each $v \in V(G)$.

A correspondence cover $\mathcal{H} = (L,H)$ is \emphdef{$q$-fold} if $|L(v)| \geq q$ for all $v \in V(G)$. The \emphdef{correspondence chromatic number} of $G$, denoted by $\chi_{c}(G)$, is the smallest $q$ such that $G$ admits a proper $\mathcal{H}$-coloring for every $q$-fold correspondence cover $\mathcal{H}$.
As correspondence coloring is a generalization of list coloring, which is a generalization of classical coloring, it follows that 
\[\chi(G) \,\leq\, \chi_{\ell}(G) \,\leq\, \chi_{c}(G).\]

A curious feature of correspondence coloring is that structural constraints can be placed on the cover graph $H$ instead of on the underlying graph $G$.
For instance, if $\mathcal{H} = (L, H)$ is a correspondence cover of a graph $G$, then $\Delta(H) \leq \Delta(G)$, so an upper bound on $\Delta(H)$ is a weaker assumption than the same upper bound on $\Delta(G)$, and there exist a number of results in which the number of available colors given to each vertex is a function of $\Delta(H)$ as opposed to $\Delta(G)$.
(This is reminiscent of the color--degree setting discussed earlier for list coloring.)
Similarly, as a result of \ref{dp:matching}, one can show the following:

\begin{proposition}[{Corollary to \cite[Proposition 1.10]{anderson2024coloring}}]\label{prop: G loc sparse implies H loc sparse}
    Let $G = (V, E)$ be a graph and let $\mathcal{H} = (L, H)$ be a correspondence cover of $G$.
    If $G$ is $k$-locally-sparse, then so is $H$.
\end{proposition}

Correspondence coloring with locally sparse cover graphs was first studied by Anderson, Kuchukova, and the author in \cite{anderson2024coloring}.
They consider a more general notion of local sparsity which is not relevant to our work.
The following can be obtained as a corollary to their main result by setting $s = t = 1$:

\begin{theorem}[Corollary to {\cite[Theorem~1.17]{anderson2024coloring}}]\label{theorem: local sparsity ADK}
    There exists a constant $d_0$ such that the following holds for all $d\geq d_0$ and $1 \leq k \leq d^{1/5}$.
    Suppose $G = (V, E)$ is a graph and $\mathcal{H} = (L, H)$ is a correspondence cover of $G$ such that
    \begin{enumerate}[label = (\textbf{S\arabic*})]
        \item $\Delta(H) \leq d$,
        \item $H$ is $k$-locally-sparse, and
        \item for every $v \in V(G)$, $|L(v)| \geq \dfrac{8\,d}{\log \left(d/\sqrt{k}\right)}$.
    \end{enumerate}
    Then, $G$ admits a proper $\mathcal{H}$-coloring.
\end{theorem}

Our main result in this setting is a correspondence coloring version of Proposition~\ref{prop: main prop}.
Note that our result extends the range of $k$ considered in Theorem~\ref{theorem: local sparsity ADK} as well as provides an improved constant factor in the lower bound of $|L(v)|$.

\begin{proposition}[Correspondence Coloring version of Proposition~\ref{prop: main prop}]\label{prop: dp coloring locally sparse}
    Let $\gamma \in (0, 1)$ be arbitrary.
    There exists a constant $d_0$ such that the following holds for all $d\geq d_0$ and $1 \leq k \leq d^{2\gamma}$.
    Suppose $G = (V, E)$ is a graph and $\mathcal{H} = (L, H)$ is a correspondence cover of $G$ such that
    \begin{enumerate}[label = (\textbf{S\arabic*})]
        \item $\Delta(H) \leq d$,
        \item $H$ is $k$-locally-sparse, and
        \item for every $v \in V(G)$, $|L(v)| \geq \dfrac{4(1+\gamma)d}{\log \left(d/\sqrt{k}\right)}$.
    \end{enumerate}
    Then, $G$ admits a proper $\mathcal{H}$-coloring.
\end{proposition}

Setting $\gamma = 1/10$ above improves the constant factor of $8$ in Theorem~\ref{theorem: local sparsity ADK} to $4.4$.
It is an open problem to improve this to $1 + o(1)$ even for $k = 1$.
With this proposition in hand, we may prove a palette sparsification result for correspondence coloring locally sparse graphs.
(Theorem~\ref{theo: main theorem} follows as a corollary.)

\begin{theorem}\label{theo: main theorem dp}
    Let $\eps > 0$ and $\gamma \in (0, 1)$ be arbitrary.
    There exists a constant $C > 0$ such that for all $\alpha \in (0, 1)$, there is $\Delta_0 \in \N$ with the following property.
    Let $\Delta \geq \Delta_0$, $n \in \N$, and let $k$ and $q \geq s$ be such that
    \[1\,\leq\, k \,\leq\, \Delta^{2\alpha\gamma}, \qquad q \defeq \dfrac{4(1+\gamma + \eps)\Delta}{\log \left(\Delta^\alpha/\sqrt{k}\right)}, \qquad \text{and} \qquad s \geq \Delta^\alpha + C\sqrt{\log n}.\]
    Let $\mathcal{H} = (L,H)$ be a $q$-fold correspondence cover of an $n$-vertex $k$-locally-sparse graph $G= (V, E)$ of maximum degree $\Delta$.
    Suppose for each vertex $v \in V$, we sample a set $S(v) \subseteq L(v)$ of size $s$ uniformly at random.
    Then, with probability at least $1 - 1/n$, $G$ has a proper $\mathcal{H}$-coloring $\phi$ with $\phi(v) \in S(v)$ for all $v \in V$.
\end{theorem}

We may actually prove the above result in the color--degree setting with the local sparsity constraint placed on the cover graph instead.

\begin{theorem}\label{theo: main theorem dp color degree}
    Let $\eps > 0$ and $\gamma \in (0, 1)$ be arbitrary.
    There exists a constant $C > 0$ such that for all $\alpha \in (0, 1)$, there is $d_0 \in \N$ with the following property.
    Let $d \geq d_0$, $n \in \N$, and let $k$ and $q \geq s$ be such that
    \[1\,\leq\, k \,\leq\, d^{2\alpha\gamma}, \qquad q \defeq \dfrac{4(1+\gamma + \eps)d}{\log \left(d^\alpha/\sqrt{k}\right)}, \qquad \text{and} \qquad s \geq d^\alpha + C\sqrt{\log n}.\]
    Let $\mathcal{H} = (L,H)$ be a $q$-fold correspondence cover of an $n$-vertex graph $G$ such that:
    \begin{enumerate}[label=\ep{\normalfont\roman*}]
        \item $H$ is $k$-locally-sparse, and
        \item $\Delta(H) \leq d$.
    \end{enumerate}
    Suppose for each vertex $v \in V$, we sample a set $S(v) \subseteq L(v)$ of size $s$ uniformly at random.
    Then, with probability at least $1 - 1/n$, $G$ has a proper $\mathcal{H}$-coloring $\phi$ with $\phi(v) \in S(v)$ for all $v \in V(G)$.
\end{theorem}

Combining the above with Proposition~\ref{prop: G loc sparse implies H loc sparse} yields Theorem~\ref{theo: main theorem dp} as a corollary.

\subsection{Probabilistic Tools}\label{subsection: prob tools}

We begin with the symmetric version of the Lov\'asz Local Lemma.

\begin{theorem}[{Lov\'asz Local Lemma; \cite[\S4]{MolloyReed}}]\label{LLL}
    Let $A_1$, $A_2$, \ldots, $A_n$ be events in a probability space. Suppose there exists $p \in [0, 1)$ such that for all $1 \leq i \leq n$ we have $\Pr[A_i] \leq p$. Further suppose that each $A_i$ is mutually independent from all but at most $d_{LLL}$ other events $A_j$, $j\neq i$ for some $d_{LLL} \in \N$. If $4pd_{LLL} \leq 1$, then with positive probability none of the events $A_1$, \ldots, $A_n$ occur.
\end{theorem}

We will also need a version of the Chernoff bound for negatively correlated random variables introduced by Panconesi and Srinivasan \cite{panconesi}. 
We say that $\set{0,1}$-valued random variables $X_1$, \ldots, $X_m$ are \emphd{negatively correlated} if for all $I \subseteq [m]$,
\[
    \E\left[\prod_{i \in I} X_i\right] \,\leq\, \prod_{i \in I} \E[X_i].
\]

\begin{theorem}[{\cite[Theorem 3.2]{panconesi}, \cite[Lemma 3]{Molloy}}]\label{lemma:NegativeChernoff}
Let $X_1$, \ldots, $X_m$ be $\set{0,1}$-valued random variables. Set $X \defeq \sum_{i=1}^mX_i$. If $X_1$, \ldots, $X_m$ are negatively correlated, then, for all $0 < t \leq \E[X]$,
\[
    \Pr[X > \E[X] + t] \,<\, \exp{\left(-\frac{t^2}{2\E{[X]}}\right)}.
\]
\end{theorem}

\section{Palette Sparsification for Locally Sparse Graphs}\label{section: palette}

In this section we derive Theorem~\ref{theo: main theorem dp color degree} from Proposition~\ref{prop: dp coloring locally sparse}. 
As mentioned in \S\ref{subsection: proof overview}, this is a variant of the argument used by Alon and Assadi in \cite[\S3.2]{alon2020palette} to prove Theorem~\ref{theorem: alon assadi} and by Anderson, Bernshteyn, and Dhawan in \cite[\S8]{anderson2022coloring} to prove Theorem~\ref{theorem: K1tt palette}.
We will employ Theorem~\ref{lemma:NegativeChernoff} to prove our result.

Let us fix $\eps > 0$ and $\gamma \in (0, 1)$ and let
\begin{align}\label{eqn: gamma' and C}
    \gamma' \defeq \frac{\eps^2}{3(1+ \gamma)}, \qquad \text{and} \qquad C \,\defeq\, 3 (\gamma')^{-3/2}.
\end{align}
We assume $\eps$ is sufficiently small for all computations in this section to hold.
Take any $\alpha \in (0, 1)$ and let $d$, $n$, $q \geq s$, and $k$ satisfy the assumptions of Theorem~\ref{theo: main theorem dp color degree}. 
That is, we assume that $d$ is sufficiently large as a function of $\eps$, $\gamma$, and $\alpha$, and we have
\[1\,\leq\, k \,\leq\, d^{2\alpha\gamma}, \qquad q \defeq \dfrac{4(1+\gamma + \eps)d}{\log \left(d^\alpha/\sqrt{k}\right)}, \qquad \text{and} \qquad s \geq d^\alpha + C\sqrt{\log n}.\]
Since $d$ is large as a function of $\eps$, $\gamma$, and $\alpha$, we may assume that $d \geq q$. 
Let $\mathcal{H} = (L,H)$ be a $q$-fold correspondence cover of an $n$-vertex graph $G$ such that:
\[\text{$H$ is $k$-locally-sparse, \qquad and \qquad $\Delta(H) \leq d$.}\]
By removing colors if necessary, we may assume $|L(v)| = q$ for each $v \in V(G)$.

Independently for each vertex $v \in V(G)$, pick a uniformly random subset $S(v) \subseteq L(v)$ of size $s$. 
Let $S \defeq \bigcup_{v \in V(G)}S(v)$ be the set of all picked colors. 
Define
\[
    S'(v) \,\defeq\, \set{c \in S(v) \,:\, |N_H(c) \cap S| \leq (1+\gamma') sd/q}.
\]
First, we will show that $S'(v)$ is large for all $v \in V(G)$ with high probability.
  
\begin{lemma}\label{lemma:Sprime}
    $\Pr\left[\exists v \in V(G),\, |S'(v)| < (1 - \gamma')s\right] \leq 1/n$.
\end{lemma}
\begin{proof}
    Take any vertex $v \in V(G)$. We will prove that
    \[
        \Pr\left[|S'(v)| < (1-\gamma')s\right] \,\leq\, \frac{1}{n^2}.
    \]
    The result then follows by taking a union bound.
    Let us start by sampling $S(v)$. 
    We may now assume that $S(v)$ is fixed. 
    Let $t \defeq \lceil \gamma' s \rceil$. 
    Call a subset $T \subseteq S(v)$ of size $t$ \emphd{bad} if
    \[
        |N_H(T) \cap S| \,>\, (1+\gamma')\frac{dts}{q}.
    \]
    Observe that if $|S'(v)| < (1-\gamma')s$, then $|S(v) \setminus S'(v)| \geq t$ and every $t$-element subset of $S(v) \setminus S'(v)$ is bad.
    Therefore, it suffices to argue that
    \[
        \Pr[\text{there is a bad set $T \subseteq S(v)$ of size $t$}] \, \leq\, \frac{1}{n^2}.
    \]
    To this end, consider an arbitrary set $T \subseteq S(v)$ of size $t$. For each color $c \in N_H(T)$, let $X_c$ be the indicator random variable of the event that $c \in S$. Define
    \[
        X \,\defeq\, \sum_{c \in N_H(T)} X_c \,=\, |N_H(T) \cap S|.
    \]
    Then,
    \[\E[X]\,=\, \sum_{c \in N_H(T)}\Pr[X_c = 1] \,=\, |N_H(T)|\,\frac{s}{q} \,\leq\,\frac{dts}{q},\]
    where we use the fact that $\Delta(H) \leq d$.
    Next, we claim that the random variables $\set{X_c \,:\, c \in N_H(T)}$ are negatively correlated. 
    Indeed, take any $I \subseteq N_H(T)$ and $c \in N_H(T) \setminus I$. Let $u \defeq L^{-1}(c)$ be the underlying vertex of $c$. 
    Then,
    \[
        \Pr\left[X_c = 1 \,\middle\vert\, \prod_{c' \in I} X_{c'} = 1\right] \,=\, \frac{s- |I \cap L(u)|}{q - |I \cap L(u)|} \,\leq\, \frac{s}{q}.
    \]
    Inductive applications of this inequality show that for any $I \subseteq N_H(T)$,
    \[
        \E\left[\prod_{c \in I} X_c\right] \,\leq\, \left(\frac{s}{q}\right)^{|I|} \,=\, \prod_{c \in I} \E[X_c],
    \]
    as desired. Using Theorem~\ref{lemma:NegativeChernoff}, we conclude that
    \[
        \Pr\left[X > (1+\gamma')\frac{d t s}{q}\right] \,<\, \exp\left(- \frac{(\gamma')^2}{3} \, \frac{dts}{q}\right) \,\leq\, \exp\left(- \frac{(\gamma')^3}{3} \, s^2\right),
    \]
    where the last step follows by definition of $t$ and since $d \geq q$.
    By the union bound, we have
    \begin{align*}
        \Pr\left[\text{there is a bad set $T \subseteq S(v)$ of size $t$}\right] \,\leq\, 2^s \, \exp\left(- \frac{(\gamma')^3}{3} \, s^2\right) \,\leq\, \exp\left(- \frac{(\gamma')^3}{4} \, s^2\right),
    \end{align*}
    where the second inequality holds for $d$ sufficiently large. Finally, since $s^2 \geq C^2\log n$, we have
    \[
        \exp\left(- \frac{(\gamma')^3}{4} \, s^2\right) \,\leq\, \exp\left(- \frac{(\gamma')^3}{4} \, C^2 \log n\right) \,\leq\, \frac{1}{n^2},
    \]
    by definition of $C$, completing the proof.
\end{proof}
  
Let $S' \defeq \bigcup_{v \in V(G)}S'(v)$ and $H' \defeq H[S']$. 
Then $\mathcal{H}' \defeq (S', H')$ is a correspondence cover of $G$ such that:
\begin{enumerate}[label=\ep{\normalfont\roman*}]
    \item $H'$ is $k$-locally-sparse,
    \item $\Delta(H') \leq d' \defeq (1+\gamma')sd/q$, and
    \item\label{item: size of S'} with probability at least $1 - 1/n$, $|S'(v)| \geq (1-\gamma')s$ for all $v \in V(G)$.
\end{enumerate}
Note that, assuming $d$ is sufficiently large, we have 
\[(1+\gamma')sd/q \,\geq\, s \,\geq\, d^{\alpha}.\]
It follows that $k \leq (d')^{2\gamma}$.
By Proposition~\ref{prop: dp coloring locally sparse}, $G$ admits a proper $\mathcal{H}'$-coloring as long as
\[|S'(v)| \,\geq\, \frac{4(1+\gamma)d'}{\log\left(d'/\sqrt{k}\right)}, \qquad \text{for each } v \in V(G).\]
To this end, we note that
\begin{align*}
    \frac{4(1+\gamma)d'}{\log\left(d'/\sqrt{k}\right)} &=
    \frac{4(1+\gamma)(1+\gamma')sd/q}{\log\left(\left((1+\gamma')sd/q\right)/\sqrt{k}\right)} \\
    &\leq \frac{4(1+\gamma)(1+\gamma')sd/q}{\log\left(d^\alpha/\sqrt{k}\right)} \\
    &= \frac{(1+\gamma)(1+\gamma')s}{(1+\gamma + \eps)} \\
    &\leq (1-\gamma')s,
\end{align*}
where the last step follows by our choice of $\gamma'$.
Therefore, by \ref{item: size of S'}, $G$ is $\mathcal{H}'$-colorable with probability at least $1 - 1/n$, and the proof of Theorem~\ref{theo: main theorem dp color degree} is complete.

\section{Proof of Proposition~\ref{prop: dp coloring locally sparse}}\label{section: wcp}

Let us first restate the proposition.

\begin{proposition*}[Restatement of Proposition~\ref{prop: dp coloring locally sparse}]
    Let $\gamma \in (0, 1)$ be arbitrary.
    There exists a constant $d_0$ such that the following holds for all $d\geq d_0$ and $1 \leq k \leq d^{2\gamma}$.
    Suppose $G = (V, E)$ is a graph and $\mathcal{H} = (L, H)$ is a correspondence cover of $G$ such that
    \begin{enumerate}[label = (\textbf{S\arabic*})]
        \item $\Delta(H) \leq d$,
        \item $H$ is $k$-locally-sparse, and
        \item for every $v \in V(G)$, $|L(v)| \geq \dfrac{4(1+\gamma)d}{\log \left(d/\sqrt{k}\right)}$.
    \end{enumerate}
    Then, $G$ admits a proper $\mathcal{H}$-coloring.
\end{proposition*}

As mentioned in \S\ref{subsection: proof overview}, we will employ the R\"odl nibble method in order to prove the above result.
We specifically employ the version of the \hyperref[algorithm: wcp]{Wasteful Coloring Procedure} generalized to correspondence coloring by Anderson, Bernshteyn, and the author in \cite{anderson2022coloring} (see \S\ref{subsection: wcp} for a detailed description).
The algorithm proceeds in stages.
At each stage, we have a graph $G$ and a correspondence cover $\mathcal{H} = (L, H)$ of $G$.
We randomly construct a partial $\mathcal{H}$-coloring $\phi$ of $G$ and lists $L'(v) \subseteq L(v)$ satisfying certain properties.
Let $G_\phi$ be the graph induced by the vertices in $G$ uncolored under $\phi$ and let $H' \defeq H\left[\bigcup_{v \in V(G_{\phi})}L'(v)\right]$.
The partial coloring $\phi$ is ``good'' if
\[\frac{\Delta(H')}{\min_{v\in V(G_\phi)}|L'(v)|} \,\leq\, \omega\,\frac{\Delta(H)}{\min_{v \in V(G)}|L(v)|},\]
for $\omega \in (0, 1)$.
If indeed the output of the \hyperref[algorithm: wcp]{Wasteful Coloring Procedure} produces a ``good'' coloring, then we may repeatedly apply the procedure until we are left with a cover $(\tilde L, \tilde H)$ of the uncolored vertices such that $\min_v |\tilde L(v)| \geq 8\Delta(\tilde H)$.
At this point, we can complete the coloring by applying the following proposition, which is a correspondence coloring version of Proposition~\ref{prop:final_blow_list} (see \cite[Appendix]{bernshteyn2019johansson} for a proof):

\begin{proposition}\label{prop:final_blow}
    Let $\mathcal{H} = (L, H)$ be a correspondence cover of $G$. 
    If there is an integer $\ell$ such that $|L(v)| \geq \ell$ and $\deg_H(c) \leq \ell/8$ for each $v \in V(G)$ and $c \in V(H)$, then $G$ admits a proper $\mathcal{H}$-coloring.
\end{proposition}

The rest of the section is structured as follows.
In \S\ref{subsection: wcp}, we will formally describe the coloring procedure mentioned.
In \S\ref{subsection: iteration}, we will prove the procedure produces a ``good'' partial coloring.
In \S\ref{subsection: recursion}, we will complete the proof of Proposition~\ref{prop: dp coloring locally sparse}.

\subsection{The Wasteful Coloring Procedure}\label{subsection: wcp}

In this section, we formally describe our coloring procedure and introduce a key lemma at the heart of our proof of Proposition~\ref{prop: main prop}.
Roughly speaking, the lemma below will allow us to, in \S\ref{subsection: recursion}, inductively color our graph until a point at which Proposition~\ref{prop:final_blow} can be applied. 
As mentioned earlier, we will utilize the \hyperref[algorithm: wcp]{Wasteful Coloring Procedure} (described formally in Algorithm~\ref{algorithm: wcp}) in order to prove the lemma.
We first introduce some notation.

For a fixed vertex $v \in V(G)$, its \emphd{average color-degree} is
\[ \overline{\deg}_\mathcal{H}(v) = \frac{1}{|L(v)|}\sum_{c\in L(v)} \deg_H(c).\]
Additionally, the list of available colors for $v$ with respect to a partial coloring $\phi$ is as follows: 
\[ L_{\phi}(v) = \set{c \in L(v) \,:\, N_H(c)\cap \im(\phi) = \0},\]
where $\im(\phi)$ is the image of $\phi$, i.e., the colors assigned to colored vertices in $G$.
The main component of the proof of Proposition~\ref{prop: dp coloring locally sparse} is the following lemma, which shows that under certain conditions on $G$ and its correspondence cover, there exists a partial coloring such that the uncolored graph has desirable properties.

\begin{lemma}\label{lemma:iteration} 
    For $\gamma \in (0, 1)$, there is $\tilde{d} \in \N$ such that the following holds. 
    Suppose $\eta$, $d$, $\ell \in \N$, and $k \in \R$ satisfy:
    \begin{enumerate}[label=\ep{\normalfont A\arabic*}]
        \item\label{item:d_large_3.1} $d$ is sufficiently large: $d \geq \tilde{d}$,
        \item\label{item: k not too large} $k$ is not too large: $1 \leq k \leq d^{2\gamma}$,
        \item\label{item:ell} $\ell$ is bounded below and above in terms of $d$: $4\eta\, d < \ell < 100d$,
        \item\label{item:eta} $\eta$ is bounded below and above in terms of $d$ and $k$: $\dfrac{1}{\log^5(d)} < \eta < \dfrac{1}{\log (d/\sqrt{k})}.$
    \end{enumerate}
    Let $G$ be a graph with a correspondence cover $\mathcal{H} = (L, H)$ such that the following holds for some $\beta$ satisfying $d^{-\gamma(1 - \sqrt{\gamma})/200} \leq \beta \leq 1/10$:
    \begin{enumerate}[label=\ep{\normalfont A\arabic*},resume]
        \item \label{item: sparsity} $H$ is $k$-locally-sparse,
        \item\label{item:Delta} $\Delta(H) \leq 2d$,
        \item\label{item:list_assumption} the list sizes are roughly between $\ell/2$ and $\ell$: $(1-\beta)\ell/2 \,\leq\, |L(v)| \leq (1+\beta)\ell$ for all $v \in V(G)$,
        \item\label{item:averaged} average color-degrees are smaller for vertices with smaller lists of colors: \[\overline{\deg}_\mathcal{H}(v) \,\leq\, \left(2 - (1 - \beta)\frac{\ell}{|L(v)|}\right)d \quad \text{for all } v\in V(G).\]
    \end{enumerate}
    Then there exists a proper partial $\mathcal{H}$-coloring $\phi$ of $G$ and 
    an assignment of subsets $L'(v) \subseteq L_\phi(v)$ to each $v \in V(G) \setminus \dom(\phi)$ with the following properties. Let
    \[
        G' \defeq G\left[V(G)\setminus \dom(\phi)\right] \qquad \text{and} \qquad H' \defeq H\left[\textstyle\bigcup_{v \in V(G')} L'(v)\right].
    \]
    Define the following quantities:
    \[
        \begin{aligned}[c]
            \keep &\defeq \left(1 - \frac{\eta}{\ell}\right)^{2d}, \\
        \uncolor &\defeq \left(1  - \frac{\eta}{\ell}\right)^{\keep\,\ell/2},
        \end{aligned}
        \qquad
        \begin{aligned}[c]
            \ell' &\defeq \keep\, \ell, \\
            d' &\defeq \keep\, \uncolor\, d,
        \end{aligned} \qquad \beta' \defeq (1+36\eta)\beta.
    \]
    Let $\mathcal{H}' \defeq (L', H')$, so $\mathcal{H}'$ is a correspondence cover of $G'$.
    Then for all $v \in V(G')$:
    \begin{enumerate}[label=\ep{\normalfont\roman*}]
        \item\label{item:I} $|L'(v)| \,\leq\, (1+\beta')\ell'$,
        
        \smallskip
        
        \item\label{item:II} $|L'(v)| \geq (1-\beta')\ell'/2$,
        
        \smallskip
        
        \item\label{item:III} $\Delta(H') \leq 2d'$,
        
        \smallskip
        
        \item\label{item:IV} $\overline{\deg}_{\mathcal{H}'}(v) \leq \left(2 - (1 - \beta')\frac{\ell'}{|L'(v)|}\right)d'.$

        \item\label{item: V} $H'$ is $k$-locally-sparse.
    \end{enumerate}
\end{lemma}

Note that condition \ref{item: V} holds automatically as $H' \subseteq H$. Also note that conditions \ref{item:I}--\ref{item:IV}  are similar to the conditions \ref{item: sparsity}--\ref{item:averaged}, except that the former uses $\beta', \eta', \ell'$. This will help us to apply Lemma~\ref{lemma:iteration} iteratively in \S\ref{subsection: recursion} to prove Proposition~\ref{prop: dp coloring locally sparse}.

Let us now describe the \hyperref[algorithm: wcp]{Wasteful Coloring Procedure} as laid out by Anderson, Bernshteyn, and the author in \cite{anderson2022coloring} for correspondence coloring (the color--degree version for list coloring appeared in \cite{alon2020palette}).

% \vspace{0.2cm}
% \begin{breakablealgorithm}
\begin{algorithm}[H]
\caption{Wasteful Coloring Procedure}\label{algorithm: wcp}
\begin{flushleft}
\textbf{Input}: A graph $G$ with a correspondence cover $\mathcal{H} = (L,H)$ and parameters $\eta \in [0,1]$ and $d$, $\ell> 0$. \\
\textbf{Output}: A proper partial $\mathcal{H}$-coloring $\phi$ and subsets $L'(v) \subseteq L_\phi(v)$ for all $v \in V(G)$.
\end{flushleft}
\begin{enumerate}[itemsep = .2cm, label = {\arabic*.}]
    \item Sample $A \subseteq V(H)$ as follows: for each $c \in V(H)$, include $c \in A$ independently with probability $\eta/\ell$.
    Call $A$ the set of \emphd{activated} colors, and let $A(v) \defeq L(v) \cap A$ for each $v \in V(G)$.

    \item Let $\{\eq(c) \,:\, c \in V(H)\}$ be a family of independent random variables with distribution
    \[\eq(c) \sim \mathsf{Ber}\left(\frac{\keep}{\left(1 - \eta/\ell\right)^{\deg_H(c)}}\right).\]
    (We discuss why this is well defined below). Call $\eq(c)$ the \emphd{equalizing coin flip} for $c$.
    
    \item Sample $K \subseteq V(H)$ as follows: for each $c \in V(H)$, include $c \in K$ if $\eq(c) = 1$ and $N_H(c) \cap A = \0$.
    Call $K$ the set of \emphd{kept} colors, and $V(H) \setminus K$ the \emphd{removed} colors. For each $v \in V(G)$, let $K(v) \defeq L(v) \cap K$.

    \item Construct $\phi : V(G) \pto V(H)$ as follows: if $A(v) \cap K(v) \neq \0$, set $\phi(v)$ to any color in $A(v) \cap K(v)$. Otherwise set $\phi(v) = \blank$.
    
    \item Call $v \in V(G)$  \emphdef{uncolored} if $\phi(v) = \blank$, and define 
    \[U \,\defeq\, \left\{c \in V(H)\,:\, \phi\left(L^{-1}(c)\right) = \blank\right\}.\]
    \ep{Recall that $L^{-1}(c)$ denotes the underlying vertex of $c$ in $G$.}

    \item\label{step: define L'} For each vertex $v \in V(G)$, let
    \[L'(v) \,\defeq\, \left\{c\in K(v)\,:\, |N_H(c) \cap K \cap U| \leq 2\,d'\right\}.\]
\end{enumerate}
\end{algorithm}
% \end{breakablealgorithm}
% \vspace{0.3cm}

Note that if $G$ and $H$ satisfy assumption \ref{item:Delta} of Lemma \ref{lemma:iteration}, i.e. $\deg_H(c) \leq 2d$, then
\[ 0 \leq \frac{\keep}{(1 - \eta/ \ell)^{\deg_H(c)}} = \left( 1 - \frac{\eta}{\ell}\right)^{2d - \deg_H(c)} \leq 1\]
and hence the equalizing coin flips are well defined. 
Furthermore, if we assume that the assumptions of Lemma \ref{lemma:iteration} hold on the input graph of the \hyperref[algorithm: wcp]{Wasteful Coloring Procedure}, then $\keep$ is precisely the probability that a color $c \in V(H)$ is kept, and $\uncolor$ is roughly an upper bound on the probability that a vertex $v \in V(G)$ is uncolored.

\subsection{Proof of Lemma~\ref{lemma:iteration}}\label{subsection: iteration}

Our proof closely follows the strategy of \cite[Lemma 3.1]{anderson2022coloring} where Anderson, Bernshteyn, and the author prove a similar result for correspondence coloring graphs with $K_{1, s, t}$-free covers.
The proof of Theorem~\ref{theorem: local sparsity ADK} adapts the above strategy in a similar manner.
In particular, Anderson, Kuchukova, and the author show that \cite[Lemma 3.1]{anderson2024coloring} differs from \cite[Lemma 3.1]{anderson2022coloring} in only a few assumptions.
Several of the lemmas used in the proof of \cite[Lemma 3.1]{anderson2022coloring} do not rely on either of these assumptions.
Therefore, the authors are able to use the proofs of these lemmas in a ``black box'' manner towards proving \cite[Lemma 3.1]{anderson2024coloring}. 
The assumptions of \cite[Lemma 3.1]{anderson2024coloring} (for $s = t = 1$) and Lemma~\ref{lemma:iteration} differ in only two ways:
\begin{enumerate}
    \item the assumption $1/2 \leq k \leq d^{2/5}$ is replaced by $1 \leq k \leq d^\gamma$, and
    \item the assumption $d^{-1/200} \leq \beta \leq 1/10$ is replaced by $d^{-\gamma(1 - \sqrt{\gamma})/200} \leq \beta \leq 1/10$.
\end{enumerate}
As $\gamma(1 - \sqrt{\gamma}) \leq 1$, the same lower bound on $\beta$ follows.
In particular, we may similarly use a number of proofs from \cite{anderson2022coloring} in a black box manner.
We therefore follow the strategy of the proof of \cite[Lemma 3.1]{anderson2022coloring}, and indeed, the notation used in this section is identical to \cite[\S4]{anderson2022coloring}.

Suppose parameters $\eta$, $d$, $\ell$, $k$, $\beta$, and a graph $G$ with a correspondence cover $\mathcal{H} = (L,H)$ satisfy the assumptions of Lemma~\ref{lemma:iteration}.
Throughout, we shall assume that $\tilde{d}$ is sufficiently large. 
We may assume that $\Delta(G) \leq 2(1+\beta) \ell d$ by removing the edges of $G$ whose corresponding matchings in $H$ are empty. 
Let the quantities $\keep$, $\uncolor$, $d'$, $\ell'$, and $\beta'$ be defined as in the statement of Lemma~\ref{lemma:iteration}. 
Suppose we have carried out the \hyperref[algorithm: wcp]{Wasteful Coloring Procedure} with these $G$ and $\mathcal{H}$. As in the statement of Lemma~\ref{lemma:iteration}, we let
\[
    G' \defeq G\left[V(G)\setminus \dom(\phi)\right],  \qquad H' \defeq H\left[\textstyle\bigcup_{v \in V(G')} L'(v)\right], \qquad \text{and} \qquad \mathcal{H}' \defeq (L', H').
\]
For each $v \in V(G)$, we define the following quantities:
\[
    \ell(v) \defeq |L(v)| \qquad \text{and} \qquad \overline{\deg}(v) \defeq \overline{\deg}_\mathcal{H}(v),
\]
as well as the following random variables:
\[
    k(v) \defeq |K(v)|, \qquad  \ell'(v) \defeq |L'(v)|, \qquad \text{and} \qquad \overline{d}(v) \defeq \frac{1}{\ell'(v)}\sum_{c\in L'(v)}|N_H(c) \cap V(H')|.
\]
By definition, if $v \in V(G')$, then $\overline{d}(v) = \overline{\deg}_{\mathcal{H}'}(v)$. Our goal is to verify that statements \ref{item:I}--\ref{item:IV} in Lemma \ref{lemma:iteration} hold for every $v \in V(G')$ with positive probability. We follow an idea of Pettie and Su \cite{PS15} (which was also used in \cite{alon2020palette, anderson2022coloring, anderson2024coloring}) by defining the following auxiliary quantities:
\begin{align*}
    \lambda(v) &\defeq \frac{\ell(v)}{\ell}, & \lambda'(v)&\defeq \frac{\ell'(v)}{\ell'}, \\
    \delta(v) &\defeq \lambda(v)\,\overline{\deg}_{\mathcal{H}}(v) + (1-\lambda(v))\,2d, & \delta'(v) &\defeq \lambda'(v)\,\overline{d}(v) + (1-\lambda'(v))\,2d'.
\end{align*}
Note that, by \ref{item:list_assumption}, we have $(1-\beta)/2 \leq \lambda(v) \leq 1+ \beta$. When $\lambda(v) \leq 1$, we can think of $\delta(v)$ as what the average color-degree of $v$ would become if we added $\ell - \ell(v)$ colors of degree $2d$ to $L(v)$. \ep{We remark that in both \cite{PS15} and \cite{alon2020palette}, the value $\lambda(v)$ is artificially capped at $1$. 
However, as noted by Anderson, Bernshteyn, and the author in \cite{anderson2022coloring}, it turns out that there is no harm in allowing $\lambda(v)$ to exceed $1$, which moreover makes the analysis simpler.} The upper bound on $\overline{\deg}(v)$ given by \ref{item:averaged} implies that
\begin{align}\label{eqn:avg_delta}
    \delta(v) \,\leq\, (1+\beta)d.
\end{align}
As demonstrated in \cite[Lemma 4.1]{anderson2022coloring}, an upper bound on $\delta'(v)$ suffices to derive statements \ref{item:II}--\ref{item:IV} in Lemma \ref{lemma:iteration}:

\begin{lemma}\label{lemma:deltaprime}
    If $\delta'(v) \leq (1 + \beta')d'$, then conditions \ref{item:II}--\ref{item:IV} of Lemma \ref{lemma:iteration} are satisfied.
\end{lemma}
\begin{proof}
    The proof is the same as the proof of \cite[Lemma 4.1]{anderson2022coloring}.
\end{proof}

As a result of the above lemma, it remains to verify there is an outcome of the \hyperref[algorithm: wcp]{Wasteful Coloring Procedure} for which the following holds:
\[\delta'(v) \leq (1 + \beta')d', \qquad \text{and} \qquad \ell'(v) \leq (1+\beta')\ell', \qquad \text{for all } v \in V(G).\]
To this end, we shall prove some intermediate results.
Before we do so, consider the following inequality, which will be useful for proving certain bounds and follows for $\tilde d$ sufficiently large:
\begin{align}\label{etabeta}
    \eta\beta \,\geq\, d^{-\gamma(1 - \sqrt{\gamma})/200}/\log^5d \,\geq\, d^{-\gamma(1 - \sqrt{\gamma})}.
\end{align}

In the next lemma, we show that $k(v)$ is concentrated around its expected value, which we then use to show condition \ref{item:I} in Lemma \ref{lemma:iteration} is satisfied with high probability.

\begin{lemma}\label{listConcentration}
    $\Pr[|k(v) -\keep\,\ell(v)| \geq \eta\beta\,\keep\,\ell(v)] \leq \exp\left(-d^{1/10}\right)$.
\end{lemma}
\begin{proof}
    The proof is the same as the proof of \cite[Lemma 4.2]{anderson2022coloring}.
\end{proof}

Since $\ell'(v)\leq k(v)$, we have the following with probability at least $1-\exp\left(-d^{1/10}\right)$:
\begin{align*}
    \ell'(v) &\leq (1+\eta\beta)\,\keep\,\ell(v) \\
    &\leq (1+\eta\beta)\,(1+\beta)\,\keep\,\ell \\
    &\leq (1+ \beta')\,\ell'.
\end{align*}
This implies that condition \ref{item:I} is met with probability at least $1 - \exp(-d^{1/10})$. 

In order to analyze the average color-degrees in the correspondence cover $\mathcal{H}'$, we define the following:
\begin{align*}
    \keptedges &\defeq \{cc'\in E(H)\,:\, c\in L(v),\ c'\in N_H(c), \text{ and } c,c'\in K\}, \\
    \uncoloredges &\defeq \{cc'\in E(H)\,:\, c\in L(v),\ c'\in N_H(c), \text{ and }
    c' \in U\}, \\
    \nd &\defeq \frac{|\keptedges \cap \uncoloredges|}{k(v)}.
\end{align*}
Note that $\nd$ is what the average color-degree of $v$ would be if instead of removing colors with too many neighbors on Step~\hyperref[step: define L']{6} of the \hyperref[algorithm: wcp]{Wasteful Coloring Procedure}, we had just set $L'(v) = K(v)$.
We refer to $E_K(v)$ as the \textit{kept edges} and $E_U(v)$ as the \textit{uncolored edges}.

The only place in our proof that relies on $H$ being $k$-locally-sparse is the following lemma, which gives a bound on the expected value of $|\keptedges|$.

\begin{lemma}\label{lemma:expectation kept}
    $\E[|\keptedges|] \leq \keep^2\,\ell(v)\,\overline{\deg}_{\mathcal{H}}(v)(1+\eta\beta)$.
\end{lemma}

\begin{proof}
    Note that 
    \[\keptedges = \sum_{c \in L(v)} \sum_{c' \in N_H(c)} \mathbbm{1}_{\{c, c' \in K\}}.\]
    If $c \in K$, then two events occur: no color in $N_H(c)$ is activated, and $c$ survives its equalizing coin flip. 
    Since activations and equalizing coin flips occur independently, we have:
    \begin{align*}
        \mathbb{E}[\keptedges] &= \sum_{c \in L(v)} \sum_{c' \in N_H(c)} \Pr[c, c' \in K]\\
        &= \sum_{c \in L(v)} \sum_{c' \in N_H(c)} \Pr[\eq(c) = 1]\, \Pr[\eq(c') = 1] \left(1 - \frac{\eta}{\ell} \right)^{|N_H(c) \cup N_H(c')|}\\
        &= \sum_{c \in L(v)} \sum_{c' \in N_H(c)} \Pr[\eq(c) = 1]\, \Pr[\eq(c') = 1] \left(1 - \frac{\eta}{\ell} \right)^{|N_H(c)|+|N_H(c')| - |N_H(c)\cap N_H(c')|}\\
        &= \keep^2 \sum_{c \in L(v)} \sum_{c' \in N_H(c)} \left(1 - \frac{\eta}{\ell} \right)^{-|N_H(c) \cap N_H(c')|}.
    \end{align*}
    For $c \in V(H)$, we will consider two cases based on the value of $\deg_H(c)$.

    \begin{enumerate}[leftmargin = \leftmargin + 1\parindent, wide, label=\textbf{(Case \arabic*)}]
        \item\label{Case1} $\deg_H(c) < d^{\sqrt{\gamma}}$. 
        Then,
        \[|N_H(c) \cap N_H(c')| \leq |N_H(c)| = \deg_H(c) \leq d^{\sqrt{\gamma}}.\]
        By the lower bound $\ell \geq 4\eta d$ from assumption \ref{item:ell}, we have
        \begin{align*}
            \sum_{c' \in N_H(c)} \left(1 - \frac{\eta}{\ell} \right)^{-|N_H(c) \cap N_H(c')|}  &\leq \sum_{c' \in N_H(c)} \left(1 - \frac{\eta}{\ell} \right)^{- d^{\sqrt{\gamma}}} \\
            &\leq \deg_H(c) \left(1 - d^{\sqrt{\gamma}} \frac{\eta}{\ell} \right)^{-1} \\
            &\leq \deg_H(c) \left( 1 - \frac{d^{- (1-\sqrt{\gamma})}}{4} \right)^{-1} \\
            &\leq \deg_H(c) \left( 1 + d^{- (1-\sqrt{\gamma})} \right),
        \end{align*}
        where the last inequality holds since $d^{-(1-\sqrt{\gamma})} \leq 1$ as $\gamma < 1$.

        \item\label{Case2} $\deg_H(c) \geq d^{\sqrt{\gamma}}$. We note that this is the only place where we use local sparsity. 
        Define the following sets:
        \begin{align*}
            \Bad &\defeq \left\{c' \in N_H(c) : |N_{H}(c') \cap N_H(c)|\geq \deg_H(c)^{\sqrt{\gamma}}\right\},\\
            \Good &\defeq N_H(c) \setminus \Bad.
        \end{align*}
        Then:
        \begin{align}
            &\hspace{1.2em}\sum_{c' \in N_H(c)} \left(1 - \frac{\eta}{\ell} \right)^{-|N_H(c) \cap N_H(c')|} \nonumber \\
            &= \sum_{c' \in \Bad} \left(1 - \frac{\eta}{\ell} \right)^{-|N_H(c) \cap N_H(c')|} + \sum_{c' \in  \Good} \left(1 - \frac{\eta}{\ell} \right)^{-|N_H(c) \cap N_H(c')|} \nonumber \\
            & \leq \sum_{c' \in \Bad} \left(1 - \frac{\eta}{\ell} \right)^{-\deg_H(c)} + \sum_{c' \in  \Good} \left(1 - \frac{\eta}{\ell} \right)^{- \deg_H(c)^{\sqrt{\gamma}}} \nonumber \\ 
            & \leq     |\Bad| \left(1 - \frac{\eta}{\ell} \right)^{-2d} + (\deg_H(c) - |\Bad|) \left(1 - \frac{\eta}{\ell} \right)^{- \deg_H(c)^{\sqrt{\gamma}}}. \label{eq:final inequality for bound with bad}
        \end{align}
        Since $\deg_H(c)^{\sqrt{\gamma}} \leq 2d$, the last expression is increasing in terms of $|\Bad|$.
        Let us now bound the size of $\Bad$.
        By \ref{item: sparsity} and the definition of local sparsity, we have
        \begin{align*}
            |\Bad|\,\deg_H(c)^{\sqrt{\gamma}} \leq 2k &\implies |\Bad| \leq \frac{2d^{2\gamma}}{\deg_H(c)^{\sqrt{\gamma}}} \\
            &\implies |\Bad| \leq 2\deg_H(c)^{2\gamma/\sqrt{\gamma} - \sqrt{\gamma}} = 2\deg_H(c)^{\sqrt{\gamma}}.
        \end{align*}
    
        With this in hand, and as $\left(1 - \frac{\eta}{\ell} \right)^{-2d} = \frac{1}{\keep} \leq 2$ due to assumption \ref{item:ell}, we have the following as a result of \eqref{eq:final inequality for bound with bad}:
        \begin{align*}
                \sum_{c' \in N_H(c)} \left(1 - \frac{\eta}{\ell} \right)^{-|N_H(c) \cap N_H(c')|}
                & \leq  2 \deg_H(c)^{\sqrt{\gamma}} +  (\deg_H(c) -  2\deg_H(c)^{\sqrt{\gamma}}) \left(1 - \frac{\eta}{\ell} \right)^{- \deg_H(c)^{\sqrt{\gamma}}} \\
                & \leq 2 \deg_H(c)^{\sqrt{\gamma}} +  \frac{(\deg_H(c) -  2\deg_H(c)^{\sqrt{\gamma}} )}{ \left(1 - \deg_H(c)^{\sqrt{\gamma}} \frac{\eta}{\ell} \right) }  \\
                & = \deg_H(c)\left(2\deg_H(c)^{-(1-\sqrt{\gamma})} +  \frac{1 -  2\deg_H(c)^{-(1-\sqrt{\gamma})} }{ \left(1 - \frac{\eta}{\ell} \deg_H(c)^{\sqrt{\gamma}}  \right) }  \right)  \\
                & \leq  \deg_H(c)\left( 2 \deg_H(c)^{-(1-\sqrt{\gamma})} +  \frac{1 -  2\deg_H(c)^{-(1-\sqrt{\gamma})} }{ \left(1 - \frac{\eta}{\ell} 2d \deg_H(c)^{-(1-\sqrt{\gamma})}  \right) }  \right) \\
                & \leq  \deg_H(c)\left( 2 \deg_H(c)^{-(1-\sqrt{\gamma})} +  \frac{1 -  2\deg_H(c)^{-(1-\sqrt{\gamma})} }{ \left(1 - \frac{1}{2} \deg_H(c)^{-(1-\sqrt{\gamma})}  \right) }  \right) \\
                & \leq \deg_H(c)\left( \deg_H(c)^{-(1-\sqrt{\gamma})\sqrt{\gamma}} + 1  \right) \\
                & \leq \deg_H(c)\left( d^{-\gamma(1-\sqrt{\gamma})} + 1  \right),
        \end{align*}
        completing this case.
    \end{enumerate}
    By \eqref{etabeta}, we have $\eta\beta \,\geq\, d^{-\gamma(1-\sqrt{\gamma})}$.
    Therefore, putting together \ref{Case1} and \ref{Case2}, we have the following:
    \begin{align*}
        \mathbb{E}[\keptedges]  &  = \keep^2 \sum_{c \in L(v)} \sum_{c' \in N_H(c)} \left(1 - \frac{\eta}{\ell} \right)^{-|N_H(c) \cap N_H(c')|} \\
        & \leq \keep^2  \left( \sum_{c \in \ref{Case1}} \deg_H(c) \left( 1 + d^{- (1-\sqrt{\gamma})} \right)     + \sum_{c \in \ref{Case2}} \deg_H(c)\left( d^{- \gamma(1-\sqrt{\gamma})} + 1  \right)  \right)  \\
        &\leq \keep^2 \sum_{c \in L(v)} \deg_H(c) \left(  d^{- \gamma (1-\sqrt{\gamma})} +1 \right) \\
        &= \keep^2 \,\ell(v)\, \overline{\deg}_{\mathcal{H}}(v)\left( d^{-\gamma(1-\sqrt{\gamma})} + 1  \right) \\
        &\leq \keep^2 \,\ell(v)\, \overline{\deg}_{\mathcal{H}}(v)\left(1 + \eta\beta\right),
    \end{align*}
    as desired.
\end{proof}

The following lemma bounds $\E[|\keptedges\cap\uncoloredges|]$ in terms of $\E[|\keptedges|]$.
It shows that roughly an $\uncolor$ fraction of the kept edges are also uncolored in expectation.

\begin{lemma}\label{keptUncolorExpectation}
    $\E[|\keptedges\cap\uncoloredges|] \leq \uncolor\,(1+4\eta\beta)\E[|\keptedges|]$.
\end{lemma}
\begin{proof}
    The proof is the same as that in \cite[Lemma 5.2]{anderson2022coloring}. 
\end{proof}

Putting together Lemmas~\ref{lemma:expectation kept} and \ref{keptUncolorExpectation}, we prove the following result:

\begin{lemma}\label{lemma: expectation kept uncolor}
    $\E[|\keptedges\cap\uncoloredges|] \leq \keep^2\,\uncolor\,\ell(v)\,\overline{\deg}_{\mathcal{H}}(v)(1+6\eta\beta)$.
\end{lemma}

\begin{proof}
    Lemmas~\ref{lemma:expectation kept} and \ref{keptUncolorExpectation} together imply that
    \begin{align*}
        \E[|\keptedges\cap\uncoloredges|] &\leq \keep^2\,\uncolor\,\ell(v)\,\overline{\deg}_{\mathcal{H}}(v)(1+4\eta\beta)(1+\eta\beta) \\
        &\leq \keep^2\,\uncolor\,\ell(v)\,\overline{\deg}_{\mathcal{H}}(v)(1+6\eta\beta). \qedhere
    \end{align*}
\end{proof}

The next lemma establishes that with high probability the quantity $|\keptedges\cap\uncoloredges|$ is not too large. 
The proof is identical to that of \cite[Lemma 4.4]{anderson2022coloring}, and thus we omit the technical details here. 
For an informal overview of the main idea, see the discussion preceding Lemma~3.5 of \cite{anderson2024coloring}.
The result of the lemma is stated below, where $d_{\max} \defeq \max\{d^{7/8}, \Delta(H)\}$.

\begin{lemma}\label{concentrationKeptUncolor}
    $\Pr[|\keptedges\cap\uncoloredges| \geq \keep^2\,\uncolor\,\ell(v)(\overline{\deg}_{\mathcal{H}}(v) + 8\eta\beta\,d_{\max})] \leq d^{-100}$.
\end{lemma}

\begin{proof}
    The proof is the same as the proof of \cite[Lemma 4.4]{anderson2022coloring}.
\end{proof}

The next two lemmas are purely computational, using the results of Lemmas~\ref{listConcentration} and \ref{concentrationKeptUncolor}.
As the calculations are identical to those in \cite[\S4]{anderson2022coloring}, we omit them here.
First, we show that $\nd$ is not too large with high probability.

\begin{lemma}\label{degreeConcentration}
    $\Pr[\nd > \keep\,\uncolor\,\overline{\deg}_{\mathcal{H}}(v) + 15\eta\beta\,\keep\,\uncolor\,d_{\max}] \leq d^{-75}$.
\end{lemma}
\begin{proof}
    The proof is the same as the proof of \cite[Lemma 4.5]{anderson2022coloring}.
\end{proof}

We can now combine the results of this section to prove the desired bound on $\delta'(v)$. Again, the proof is purely computational, and so we refer to \cite[Lemma 4.6]{anderson2022coloring}.

\begin{lemma}\label{deltaBound}$\Pr\left[\delta'(v) \leq (1+\beta')d'\right] \geq 1-d^{-50}$.
\end{lemma}

\begin{proof}
The proof is the same as the proof of \cite[Lemma 4.6]{anderson2022coloring}.
\end{proof}

We are now ready to finish the proof of Lemma \ref{lemma:iteration}.

\begin{proof}[Proof of Lemma \ref{lemma:iteration}]
    This proof is the same as the proof of \cite[Lemma 3.1]{anderson2022coloring}, but is included here for completeness. We perform the \hyperref[algorithm: wcp]{Wasteful Coloring Procedure} on $G$ and $\mathcal{H}$ and define the following random events for each $v \in V(G)$: 
    \begin{enumerate}
        \item $A_v \defeq \set{\ell'(v) \leq (1+\beta')}$,
        \item $B_v \defeq \set{\delta'(v) \geq (1+\beta')d'}$.
    \end{enumerate}
    We now use the \hyperref[LLL]{Lov\'asz Local Lemma}. 
    By Lemmas \ref{listConcentration} and \ref{deltaBound}, we have:
    \[\Pr[A_v] \leq \exp\left(-d^{1/10}\right) \leq d^{-50}, \qquad
    \Pr[B_v] \leq d^{-50}.\]
    Let $p \defeq d^{-50}$.
    Note that the events $A_v$, $B_v$ are mutually independent from the events of the form $A_u$, $B_u$, where $u \in V(G)$ is at distance more than $4$ from $v$. Since $\Delta(G) \leq 2(1+\beta)\ell d$, there are at most $2(2(1+\beta)\ell d)^4 \leq d^{10}$ events corresponding to the vertices at distance at most $4$ from $v$. So we let $\dlll \defeq d^{10}$ and observe that $4p\dlll = 4 d^{-40} < 1$. By the \hyperref[LLL]{Lov\'asz Local Lemma}, with positive probability none of the events $A_v$, $B_v$ occur. 
    By Lemma~\ref{lemma:deltaprime}, this implies that, with positive probability, the output of the \hyperref[algorithm: wcp]{Wasteful Coloring Procedure} satisfies the conclusion of Lemma~\ref{lemma:iteration}.
\end{proof}

\subsection{The Recursion}\label{subsection: recursion}

Let $\gamma$, $d$, $k$ be as defined in the statement of the proposition, and let $\eps > 0$ be a sufficiently small constant in terms of $\gamma$.
Define $\gamma' \in (\gamma, 1)$ as follows:
\begin{align}\label{eqn: gamma'}
    \gamma' \defeq \gamma\left(\frac{1 + \gamma - 7\eps/32}{2\gamma - \eps/4}\right).
\end{align}
Let us first verify that $\gamma' \in (\gamma, 1)$.
Note that
\begin{align*}
    \gamma' > \gamma &\iff 1 + \gamma - 7\eps/32 > 2\gamma - \eps/4 \\
    &\iff 1 + \eps/32 > \gamma,
\end{align*}
which is always true as $\gamma < 1$.
% It is clear that $\gamma' > 0$ for sufficiently small $\eps$.
Let us now show $\gamma' < 1$.
We have:
\begin{align*}
    \gamma' < 1 &\iff \frac{\gamma + \gamma^2 - 7\gamma\eps/32}{2\gamma - \eps/4} < 1 \\
    &\iff \gamma + \gamma^2 - 7\gamma\eps/32 < 2\gamma - \eps/4 \\
    &\iff \eps < \frac{4\gamma(1 - \gamma)}{1 - 7\gamma/8},
\end{align*}
which follows for sufficiently small $\eps$.

Recall that our coloring procedure is iterative.
To prove Proposition~\ref{prop: dp coloring locally sparse}, we start by defining several parameters:
\begin{align*}
    C &\defeq 4(1+\gamma)& &\\
    \mu &\defeq \frac{C-\epsilon}{2}\log\left(1+\frac{\epsilon}{8C}\right)  &\eta &\defeq \mu/\log \left(d/\sqrt{k}\right)\\
    \ell_0 &\defeq C\,d/\log\left(d/\sqrt{k}\right) & d_0 &\defeq d \\
    \keep_i &\defeq \left(1 - \frac{\eta}{\ell_i}\right)^{2d_i}
    & \uncolor_i &\defeq \left(1 - \frac{\eta}{\ell_i}\right)^{\keep_i\,\ell_i/2}
    \\
    \ell_{i+1} &\defeq \keep_i\, \ell_i & d_{i+1}&\defeq \keep_i\, \uncolor_i\, d_i \\
    \beta_0 &\defeq d_0^{-\gamma'(1 - \sqrt{\gamma'})/200} & \beta_{i+1} &\defeq \max\left\{(1+36\eta)\beta_i,\, d_{i+1}^{-\gamma'(1 - \sqrt{\gamma'})/200}\right\}.
\end{align*}
We begin with the following lemma which lists some of the main relations between the above parameters.

\begin{lemma}\label{lemma: iteration helper 1}
    % Let $0 \leq \gamma < \gamma' < 1$. 
    There exists $d^\star$ such that whenever $d \geq d^\star$ and $1 \leq k \leq d^{2\gamma}$, the following hold:
    \begin{enumerate}[label=(R\arabic*)]
        \item\label{item: ratio decreases} For every $i \in \N$, $d_i/\ell_i \leq d_1/\ell_1 \leq \log\left(d/\sqrt{k}\right)/C$.
        \item\label{item: list lower bound} For every $i \in \N$, $\ell_i \geq d\left(d/\sqrt{k}\right)^{4/(C - 7\eps/8)}$.
        \item\label{item: i star} There exists a minimal integer $i^\star \,\leq\, \left\lceil\frac{16}{\mu}\log \left(d/\sqrt{k}\right)\log\log \left(d/\sqrt{k}\right)\right\rceil$ such that $d_{i^\star}\,\leq\,\ell_{i^\star}/100$.
    \end{enumerate}
\end{lemma}

\begin{proof}
    The proof is identical to that of \cite[Lemma~4.1]{anderson2024coloring} for $s = t = 1$, \textit{mutatis mutandis}.
    We note that our bounds on $k$ and $C$ are weaker than theirs, however, it is enough to have $k = o(d^2)$ and $C \geq 4 + \eps$ for their arguments.
    As $\gamma \in (0, 1)$, these do indeed hold for sufficiently small $\eps$.
\end{proof}

\begin{lemma}\label{lemma: iteration helper 2}
    % Let $\epsilon >0$ be sufficiently small and let $0 \leq \gamma < \gamma' < 1$. 
    There exists $d^\#$ such that whenever 
    \[d \geq d^\#, \quad \text{and} \quad 1 \,\leq\, k \,\leq\, d^{2\gamma},\]
    the following hold for all $i \in \N$ with $0 \leq i < i^\star$, where $i^\star$ is defined (and guaranteed to exist) by \ref{item: i star}.
    \begin{enumerate}[label=\ep{I\arabic*}, leftmargin = \leftmargin + \parindent]
        \item\label{item: di not too large} $d_i$ is sufficiently large: $d_i \geq \tilde{d}$, where $\tilde{d}$ is from Lemma~\ref{lemma:iteration},
        \item\label{item: k not too large} $k$ is not too large: $k \leq d_i^{2\gamma'}$.
        \item\label{item: ell bounded by d} $\ell_i$ is bounded below and above in terms of $d_i$: $4\eta\,d_i < \ell_i < 100d_i$,
        \item\label{item: eta small} $\eta$ is sufficiently small: $\dfrac{1}{\log^5d_i} < \eta < \dfrac{1}{\log (d_i/\sqrt{k})}$,
        \item\label{item: beta small} $\beta_i$ is small: $\beta_i \leq 1/10$. \ep{Note that the bound $\beta_i \geq d_i^{-\gamma'(1 - \sqrt{\gamma'})/200}$ holds by definition.}
    \end{enumerate}
\end{lemma}

\begin{proof}
    For all $i < i^\star$, we have $d_i \geq \ell_i/100$, and by \ref{item: list lower bound} of Lemma~\ref{lemma: iteration helper 1}, it follows for $d$ large enough that $\ell_i \geq d\left(d/\sqrt{k}\right)^{-4/(C-7\epsilon/8)}$. 
    Note that $k \geq 1$ and $C \geq 4 + \eps$ for $\eps$ sufficiently small.
    Therefore, we have the following for $i < i^\star$:
    \begin{equation}\label{eq: di bounded by d to epsilon over 40}
        d_i \geq \frac{1}{100}d\left(d/\sqrt{k}\right)^{-4/(C-7\epsilon/8)} \geq d^{\epsilon/40}. \stepcounter{equation}\tag{\theequation}
    \end{equation}
    Setting $d^\# \geq \tilde d^{40/\eps}$ completes the proof of \ref{item: di not too large}.

    Let us now consider \ref{item: k not too large}.
    Since we assume $i < i^\star$, we have the following as a result of \ref{item: list lower bound}:
    \begin{align*}
        d_i &\geq \frac{1}{100}\,d\left(d/\sqrt{k}\right)^{-4/(C-7\epsilon/8)}\\
        & = \frac{1}{100}\,d^{\left(1-\frac{4}{C-7\epsilon/8}\right)}\,k^{\left(\frac{2}{(C-7\epsilon/8)} \right)}\\
        & \geq d^{\left(1-\frac{4(1+\eps/32)}{C-7\epsilon/8}\right)}k^{\left(\frac{2}{(C-7\epsilon/8)} \right)} \qquad \qquad \left(\text{for $d$ sufficiently large}\right) \\
        & \geq  k^{\frac{1}{2\gamma}\left(1-\frac{4(1+\eps/32)}{C-7\epsilon/8}\right)+\frac{2}{(C-7\epsilon/8)}}  \qquad \qquad \left(\text{as } k \leq d^{2\gamma}\right) \\
        & = k^{\frac{1}{2\gamma}\left(1-\frac{4(1+\eps/32 - \gamma)}{C-7\epsilon/8}\right)}.
    \end{align*}
    Let us consider the exponent above.
    Plugging in the value of $C$, we have
    \begin{align*}
        1-\frac{4(1+\eps/32 - \gamma)}{C-7\epsilon/8} = 1-\frac{(1+\eps/32 - \gamma)}{1 + \gamma -7\epsilon/32} = \frac{2\gamma - \eps/32 - 7\eps/32}{1 + \gamma - 7\epsilon/32} = \frac{\gamma}{\gamma'},
    \end{align*}
    where the last step follows by definition of $\gamma'$ (see \eqref{eqn: gamma'}). 
    This completes the proof of \ref{item: k not too large}.

    By \ref{item: ratio decreases}, we have
    \[\frac{\ell_i}{d_i} \geq \frac{\ell_0}{d_0} = \frac{C\,\eta}{\mu}.\]
    Since $\mu < 1$ for $\epsilon$ sufficiently small, this shows 
    \[\ell_i \,\geq\, \frac{C\,\eta}{\mu}d_i \,\geq\, C\,\eta d_i \,\geq\, 4\eta d_i,\]
    proving \ref{item: ell bounded by d} as the upper bound always holds for $i < i^\star$.

    Using the fact that $d/\sqrt{k} \leq d$, we see for $d$ large enough that
    \[\frac{1}{\log^5(d_i)} \,\leq\, \frac{1}{\log^5(d^{\epsilon/40})} \,=\, \frac{(40/\epsilon)^5}{\log^5(d)} < \frac{\mu}{\log(d)} \,\leq\, \eta.\]
    Furthermore, for $\epsilon$ sufficiently small, we have $\mu < 1$. Thus, as $d_i \leq d$, we may conclude that
    \[\eta \,\leq\, \frac{1}{\log\left(d/\sqrt{k}\right)} \,\leq\, \frac{1}{\log\left(d/\sqrt{k}\right)},\]
    completing the proof of \ref{item: eta small}.

    Let $0\leq i' \leq i^\star-1$ be the largest integer such that
    $\beta_{i'} = d_{i'}^{-\gamma'(1 - \sqrt{\gamma'})/200}$.
    Then 
    \begin{align*}
        \beta_{i^\star-1} &= (1 + 36\eta)^{i^\star-1 - i'} d_{i'}^{-\gamma'(1 - \sqrt{\gamma'})/200} \\
        &\leq (1+36\eta)^{i^\star}\left(d^{\epsilon/40}\right)^{-\gamma'(1 - \sqrt{\gamma'})/200} \\
        &\leq \exp\left(36\eta i^\star\right)d^{-\epsilon\gamma'(1 - \sqrt{\gamma'})/8000}.
    \end{align*}
    By Lemma \ref{lemma: iteration helper 1}\ref{item: i star}, we may conclude the following:
    \begin{align*}
        \beta_{i^\star-1} & \leq \exp\left(36\eta i^\star\right)d^{-\epsilon\gamma'(1 - \sqrt{\gamma'})/8000} \\
        & \leq \exp\left(576\,\frac{\eta}{\mu}\log \left(d/\sqrt{k}\right)\log\log \left(d/\sqrt{k}\right) \right)d^{-\epsilon\gamma'(1 - \sqrt{\gamma'})/8000} \\
        & = \left(\log\left(d/\sqrt{k}\right)\right)^{576}d^{-\epsilon\gamma'(1 - \sqrt{\gamma'})/8000} \\
       & \leq \frac{1}{10},
    \end{align*}
    for $d$ large enough as $k \geq 1$. This shows \ref{item: beta small} and completes the proof.
\end{proof}

We are now ready to prove Proposition~\ref{prop: dp coloring locally sparse}.

\begin{proof}[Proof of Proposition~\ref{prop: dp coloring locally sparse}]
    Let $\gamma \in (0, 1)$ and let $\gamma'$ be as defined in \eqref{eqn: gamma'}. 
    We will not explicitly compute $d_0$ but simply take $d$ large enough when needed. 
    Thus let $d$ be sufficiently large so that we may apply Lemma \ref{lemma: iteration helper 2}. 

    Let $G$ be a graph and $\mathcal{H} = (L, H)$ be a correspondence cover of $G$ satisfying the hypotheses of the proposition.
    Set
    \begin{align*}
        G_0 \defeq G, \qquad \mathcal{H}_0 = (L_0, H_0) \defeq \mathcal{H}
    \end{align*}
    By removing some of the vertices from $H$ if necessary, we may assume that $|L(v)| = \ell_0$ for all $v \in V(G)$. 
    
    At this point, we recursively define a sequence of graphs $G_i$ with correspondence cover $\mathcal{H}_i = (L_i, H_i)$ for $0 \leq i \leq i^\star$ by iteratively applying Lemma \ref{lemma:iteration} as follows: Lemma \ref{lemma: iteration helper 2} shows conditions \ref{item:d_large_3.1}--\ref{item:eta} of Lemma \ref{lemma:iteration} (with $d_i$ in place of $d$, $\ell_i$ in place of $\ell$, and $\gamma'$ in place of $\gamma$) hold for each $0 \leq i \leq i^\star-1$. 
    Therefore, if $G_i$ is a graph with correspondence cover $\mathcal{H}_i = (L_i, H_i)$ such that for $\beta_i$ as defined at the start of the section,
    \begin{enumerate}[label=\ep{\normalfont\arabic*}]
        \item\label{item:6} $H_i$ is $k$-locally-sparse,
        \item\label{item:7} $\Delta(H_i) \leq 2d_i$,
        \item\label{item:8} the list sizes are roughly between $\ell_i/2$ and $\ell_i$: \[(1-\beta_i)\ell_i/2 \,\leq\, |L_i(v)| \,\leq\, (1+\beta_i)\ell_i \quad \text{for all } v \in V(G_i),\]
        \item\label{item:9} average color-degrees are smaller for vertices with smaller lists of colors: \[\overline{\deg}_{\mathcal{H}_i}(v) \,\leq\, \left(2 - (1 - \beta_i)\frac{\ell_i}{|L_i(v)|}\right)d_i \quad \text{for all } v\in V(G_i),\]
    \end{enumerate}
    then $G_i$ and $\mathcal{H}_i$ satisfy the required hypotheses of Lemma~\ref{lemma:iteration}.
    In particular, there exists a partial $\mathcal{H}_i$-coloring $\phi_i$ of $G_i$ and an assignment of subsets $L_{i+1}(v) \subseteq (L_i)_{\phi_i}(v)$ to each vertex $v \in V(G_i) \setminus \dom(\phi_i)$ such that, setting
    \[
        G_{i+1} \defeq G_i[V(G_i) \setminus \dom(\phi_i)], \qquad H_{i+1} \defeq H_i \left[\textstyle\bigcup_{v \in V(G_{i+1})} L_{i+1}(v)\right],
    \]
    \[
        \text{and} \qquad \mathcal{H}_{i+1} \defeq (L_{i+1}, H_{i+1}),
    \]
    we have items \ref{item:6}--\ref{item:9} hold for $G_{i+1}$ and $\mathcal H_{i+1}$ with parameter $\beta_{i+1}$.
    Thus we may iteratively apply Lemma \ref{lemma:iteration} $i^\star$ times, starting with $G_0$ and $\mathcal{H}_0$, to define $G_i$ and $\mathcal{H}_i$ for $0 \leq i \leq i^\star$.
    
    We now show $G_{i^\star}$ and $\mathcal{H}_{i^\star}$ satisfy the hypotheses of Proposition \ref{prop:final_blow}.
    By \ref{item:8}, it follows for each $v \in V(G_{i^\star})$,
    \[|L_{i^\star}(v)| \geq (1-\beta_{i^\star})\ell_{i^\star}/2\geq (1-(1+36\eta)\beta_{i^\star-1})\ell_{i^\star}/2.\]
    
    Since $\beta_{i^\star-1} \leq \frac{1}{10}$ and $\eta \leq \frac{1}{100}$ for $d$ large enough, it follows 
    \[(1-(1+36\eta)\beta_{i^\star-1})\ell_{i^\star}/2 \,\geq\, \left(1-\frac{1+36/100}{10}\right)\ell_{i^\star}/2 \,\geq\, \frac{1}{4}\ell_{i^\star} \]
    for $d$ large enough.
    As $\ell_{i^\star} \geq 100d_{i^\star}$, we have:
    \[\frac{1}{4}\ell_{i^\star} \geq 25d_{i^\star}.\]
    We have by \ref{item:7} that $\Delta(H_{i^\star}) \leq 2d_{i^\star}$, thus for all $c \in V(H_{i^\star})$, we have 
    \[ 25d_{i^\star} \geq 10\deg_{H_{i^\star}}(c).\]
    Putting this chain of inequalities together yields $|L_{i^\star}(v)| \geq 8 \deg_{H_{i^\star}}(c)$, as desired. 
    
    Therefore by Proposition \ref{prop:final_blow} there exists an $\mathcal{H}_{i^\star}$-coloring of $G_{i^\star}$, say $\phi_{i^\star}$.
    Letting $\phi_i$ be  the partial $\mathcal{H}_{i}$-colorings of $G_i$ for $0 \leq i \leq i^\star-1$ guaranteed at each application of Lemma \ref{lemma:iteration}, it follows \[\bigcup_{i = 0}^{i^\star}\phi_i\]
    is a proper $\mathcal{H}$-coloring of $G$, as desired.
\end{proof}

\printbibliography

\end{document}